\documentclass[sigconf]{acmart}

\usepackage[linesnumbered,ruled,vlined]{algorithm2e}
\usepackage{amsthm}
\usepackage{multirow}
\usepackage{enumitem}

\definecolor{ColorRoundE}{rgb}{0.85,0.03,0.19}
\definecolor{ColorRoundD}{rgb}{0.376,0.619,0.843}
\definecolor{ColorRoundF}{rgb}{0.54,0.17,0.87}
\newcommand{\ChengSIGMOD}{}

\newcommand{\roundA}{}
\newcommand{\roundB}{}
\newcommand{\revision}{}
\newcommand{\LC}{}
\def\CR{1}

\newtheorem{problem}{Problem}

\newtheorem{theorem}{Theorem}
\newtheorem{definition}{Definition}
\newtheorem{lemma}{Lemma}
\newtheorem{example}{Example}
\newtheorem{observation}{Observation}

\AtBeginDocument{%
  \providecommand\BibTeX{{%
    \normalfont B\kern-0.5em{\scshape i\kern-0.25em b}\kern-0.8em\TeX}}}

\copyrightyear{2023}
\acmYear{2023}
\setcopyright{acmlicensed}\acmConference[SIGMOD '23]{Proceedings of the 2023 International Conference on Management of Data}{June 18--23, 2023}{Seattle, WA, USA}
\acmBooktitle{Proceedings of the 2023 International Conference on Management of Data (SIGMOD '23), June 18--23, 2023, Seattle, WA, USA}
\acmDOI{XXXXXXX.XXXXXXX}

\acmPrice{15.00}
\acmISBN{978-1-4503-XXXX-X/18/06}

\begin{document}
\sloppy
\title{Maximum $k$-Biplex Search on Bipartite Graphs: A Symmetric-BK Branching Approach}

\author{Kaiqiang Yu}
\affiliation{%
  \institution{Nanyang Technological University}
  \country{Singapore}
}
\email{kaiqiang002@e.ntu.edu.sg}
\author{Cheng Long}
\authornotemark[1]
\thanks{*Cheng Long is the corresponding author.}
\affiliation{%
  \institution{Nanyang Technological University}
  \country{Singapore}
}
\email{c.long@ntu.edu.sg}


\begin{abstract}
Enumerating maximal $k$-biplexes (MBPs) of a bipartite graph has been used for applications such as fraud detection. Nevertheless, there usually exists an exponential number of MBPs, which brings up two issues when enumerating MBPs, namely the effectiveness issue (many MBPs are of low values) and the efficiency issue (enumerating all MBPs is not affordable on large graphs). Existing proposals of tackling this problem impose constraints on the number of vertices of each MBP to be enumerated, yet they are still not sufficient (e.g., they require to specify the constraints, which is often not user-friendly, and cannot control the number of MBPs to be enumerated directly). Therefore, in this paper, we {\LC study the problem of finding} $K$ MBPs with the most edges called MaxBPs, where $K$ is a positive integral user parameter. The new proposal well avoids the drawbacks of existing proposals (i.e., the number of MBPs to be enumerated is directly controlled and the MBPs to be enumerated tend to have high values since they have more edges than the majority of MBPs). We formally prove the NP-hardness of the problem. 
We then design two \emph{branch-and-bound} algorithms, among which, the better one called \texttt{FastBB} improves the worst-case time complexity to $O^*(\gamma_k^ n)$, where $O^*$ suppresses the polynomials, $\gamma_k$ is {a} real number that relies on $k$ and is \emph{strictly} smaller than 2, and $n$ is the number of vertices in the graph. 
For example, for $k=1$, $\gamma_k$ is equal to $1.754$. 
We further introduce three techniques for boosting the performance of the branch-and-bound algorithms, among which, the best one called \texttt{PBIE} can further improve the time complexity to $O^*(\gamma_k^{d^3})$ for large sparse graphs, where $d$ is the maximum degree of the graph (note that $d<< n$ for sparse graphs). 
We conduct extensive experiments on both real and synthetic datasets, and the results show that our algorithm is up to four orders of magnitude faster than all baselines and finding MaxBPs works better than finding all MBPs for a fraud detection application.
\end{abstract}



\maketitle

\section{introduction}
\label{sec:intro}
Bipartite graph is an important type of data structure where vertices are divided into two disjoint sets at two sides and each edge connects a vertex at one side and another at the other side. Bipartite graphs have been widely used for modeling the interactions between two types of entities in many real applications, with the entities being vertices and the interactions being edges. 
Some examples include commenting interactions between users and articles in social media~\cite{zhang2020overview}, purchasing interactions between customers and products in E-commerce~\cite{wang2006unifying}, visiting interactions between users and websites in web applications~\cite{beutel2013copycatch}, etc.

A bipartite subgraph is called a $k$-biplex if each vertex at one side \emph{disconnects} at most $k$ vertices at the other side, where $k$ is often a small positive integer~\cite{DBLP:journals/sadm/SimLGL09,yu2021efficient,yu2022kbiplex}. 
$k$-biplex has been used to capture \emph{dense/cohesive} subgraphs in a given bipartite graph for solving practical problems such as anomaly detection~\cite{gangireddy2020unsupervised,yu2021graph}, online recommendation~\cite{DBLP:conf/kdd/PoernomoG09a, DBLP:conf/icdm/GunnemannMRS11} and community search~\cite{DBLP:conf/ssdbm/HaoZW020,DBLP:journals/corr/abs-2011-08399}. For example, in e-commerce platforms, some agents are paid to promote the ranks of certain products by coordinating a group of fake users to post fake comments. The subgraphs induced by these fake users and the products they promote would be dense and likely to be $k$-biplexes~\cite{gangireddy2020unsupervised}. 

{\roundA
\smallskip
\noindent\textbf{Motivations.} 
There are a few studies on the problem of enumerating \emph{maximal $k$-biplexes} (MBPs)~\cite{DBLP:journals/sadm/SimLGL09,yu2021efficient,yu2022kbiplex}.
Nevertheless, there usually exist an exponential number of MBPs, which brings two issues of enumerating all MBPs. First, not all MBPs carry essential information (e.g., those MBPs with few vertices are often deemed not interesting~\cite{DBLP:journals/sadm/SimLGL09}). Second, the process of enumerating all MBPs is costly (e.g., according to existing studies~\cite{yu2021efficient}, 
enumerating all MBPs {\LC on large graphs} is not affordable). To mitigate these issues, existing studies~\cite{yu2022kbiplex} impose some constraints on the number of vertices at each of the two sides of a MBP to be enumerated, e.g., they enumerate only MBPs with \emph{at least} a certain number vertices at each side (which we call large MBPs). While this strategy makes it possible to control the number of MBPs to be enumerated, it achieves the goal only in an indirect way and introduces an additional issue of requiring to set proper thresholds of the number of vertices. In cases where users have no prior knowledge about the thresholds, they would find it not user-friendly. They can try different thresholds, but then the enumeration processes would run multiple times and bring up the time costs. 

Motivated by these issues, in this paper, we {\LC study} the problem of finding $K$ MBPs with the most edges among large MBPs, where $K$ is an integral parameter. We call each of these MBPs a \emph{maximum $k$-biplex} (MaxBP). Compared with existing studies on $k$-biplexes, the problem of finding MaxBPs enjoys three advantages. First, each MaxBP to be found has more edges than those that are not returned, and thus the found MBP is of more significance. In our experiments, we verify this via a case study, which shows that a method based on MaxBPs provides F1 score up to 0.99 for a fraud detection task. Second, it provides a \emph{direct} control on the number of MBPs to be found without the need of making multiple trials of enumeration. Third, compared with alternative proposals, e.g., finding the \emph{first} $K$ MBPs (as existing studies~\cite{yu2022kbiplex} do) or finding $K$ MBPs with the most vertices, our solution to the problem of finding MaxBPs would return MBPs that are more balanced, which are deemed to be superior over imbalanced structures~\cite{lyu2020maximum}. 
}

{\roundA
\smallskip
\noindent\textbf{Baseline Methods.} 
Given the fact that MaxBPs are the MBPs with the most edges, we can adapt the existing algorithms of enumerating large MBPs, namely \texttt{iMB}~\cite{yu2021efficient} and \texttt{iTraversal}~\cite{yu2022kbiplex}, to find the MaxBPs and yield the following two baseline methods. 
The first one called \texttt{iMBadp} adapts \texttt{iMB}~\cite{yu2021efficient} (a branch-and-bound algorithm) by incorporating additional pruning techniques which prune those branches that cannot hold any MBP with more edges than the $K$ MaxBPs found so far. 
The second baseline called \texttt{iTradp} simply runs  \texttt{iTraversal}~\cite{yu2022kbiplex} (a reverse search based algorithm) and returns $K$ MaxBPs since \texttt{iTraversal} cannot be equipped with additional pruning techniques easily.
%
Nevertheless, \texttt{iMBadp} has its efficiency highly rely on the pruning techniques and \texttt{iTradp} needs to explore all large MBPs, which is time-consuming. Both of them have the worst-case time complexity of $O^*(2^{n})$, where $n$ is the number of vertices in the given bipartite graph and $O^*$ suppresses polynomials. 
%
Furthermore, we can adapt those existing algorithms of enumerating large maximal $k$-plexes, which are counterparts of MBPs on general graphs. Here, a $k$-plex is a subgraph with each vertex disconnecting at most $k$ vertices in the subgraph. 
Therefore, the third baseline method called \texttt{FPadp} adapts \texttt{FaPlexen}~\cite{DBLP:conf/aaai/ZhouXGXJ20} (the algorithm for enumerating large maximal $k$-plexes with the number of vertices at least a threshold) with some additional pruning techniques that are tailored for MBPs. 
Still, \texttt{FPadp} is inferior to the new algorithms we will introduce in this paper both theoretically and empirically.
}

{\roundA
\smallskip
\noindent\textbf{New Methods.}
We first introduce a \emph{branch-and-bound} algorithm called \texttt{BasicBB}, which is based on a conventional and widely-used branching strategy that we call \emph{Bron-Kerbosch (BK) branching}~\cite{bron1973algorithm}. The BK branching recursively partitions the search space (i.e., the set of all possible MBPs) to multiple sub-spaces via \emph{branching}. \texttt{BasicBB} has the worst-case time complexity $O(n\cdot d\cdot 2^{n})$ (i.e., $O^*(2^{n})$),
where $d$ is the maximum degree of the graph.
This time complexity is the same as those of the baseline methods. 
We then introduce a new branching strategy called \emph{Symmetric-BK (Sym-BK) branching}, which is symmetric to the BK branching but better suits our problem of finding MaxBPs. We further present our method for determining an ordering of vertices, which is critical for Sym-BK branching. We finally introduce a new branch-and-bound algorithm called \texttt{FastBB}, which is based on the Sym-BK branching. \texttt{FastBB} has its worst-case time complexity $O(n\cdot d \cdot \gamma_k^{n})$ (i.e., $O^*(\gamma_k^{n})$), where $\gamma_k$ is \emph{strictly} smaller than 2 and depends on the setting of $k$. For example, when $k=1$, $\gamma_k$ is $1.754$. 
This is a remarkable theoretical improvement over the prior solutions given that many existing algorithms of enumerating subgraphs are based on BK branching and have the worst-case time complexity of $O^*(2^{n})$~\cite{DBLP:journals/sadm/SimLGL09,yu2021efficient,wang2017parallelizing}.

{\ChengSIGMOD{In addition, we adapt two existing techniques for boosting the efficiency and scalability of the branch-and-bound (\texttt{BB}) algorithms including \texttt{BasicBB} and \texttt{FastBB}.}} They share the idea of constructing \emph{multiple} problem instances of finding MaxBPs each on a smaller subgraph. Specifically, the first technique, called \emph{progressive bounding} (\texttt{PB}), is adapted from an existing study of finding the biclique with the most edges~\cite{lyu2020maximum}. The second technique, called \emph{inclusion-exclusion} (\texttt{IE}), is adapted from the \emph{decomposition} technique, which has been widely used for enumerating and finding subgraph structures~\cite{wang2022listing,chen2021efficient,DBLP:conf/kdd/ConteMSGMV18}. \texttt{PB} improves the practical performance of a \texttt{BB} algorithm only while \texttt{IE} improves both the theoretical time complexity (for certain sparse graphs) and the practical performance.
{\ChengSIGMOD We then present \texttt{PBIE}, which combines PB and IE naturally.} \texttt{PBIE} enjoys the benefits of both \texttt{PB} and \texttt{IE}. We note that all these techniques are orthogonal to the \texttt{BB} algorithms, i.e., any of these techniques, namely \texttt{PB}, \texttt{IE}, and \texttt{PBIE}, can be used to boost the efficiency and/or scalability of \texttt{BasicBB} and \texttt{FastBB}. 
To the best of our knowledge, this is first time that \texttt{PB} and \texttt{IE} are combined naturally.

}

\smallskip
\noindent\textbf{Contributions.} Our major contributions are summarized below.
\begin{itemize}[leftmargin=*]
    \item We study the problem of finding MaxBPs, and formally prove the NP-hardness of the problem.
    \item We propose an efficient branch-and-bound algorithm, called \texttt{FastBB}, which is based on a novel Sym-BK branching strategy. In particular, \texttt{FastBB} achieves the state-of-the-art worst-case time complexity $O(n\cdot d\cdot\gamma_k^{n})$ with $\gamma_k < 2$.
    \item We further introduce a combined framework, called \texttt{PBIE}, to further boost the performance of \texttt{FastBB}. \texttt{PBIE} combines two adapted frameworks, namely the progressive bounding framework \texttt{PB} and the inclusion-exclusion based framework \texttt{IE}. 
    When \texttt{PBIE} is used with \texttt{FastBB}, the worst-time time complexity becomes $O(d^4\cdot \gamma_k^{d^3})$. Note that this is better than that of \texttt{FastBB} on certain graphs (e.g., those sparse graphs with $d << n$).
    \item We conduct extensive experiments using both real and synthetic datasets, and the results show that (1) the proposed algorithms are up to four orders of magnitude faster than all baselines and (2) finding MaxBPs work better in a fraud detection task than enumerating MBPs.
\end{itemize}

\noindent\textbf{Roadmap.} The rest of this paper is organized below. Section~\ref{sec:problem} defines the problem and shows its NP-hardness. Section~\ref{sec:alg} presents the branch-and-bound algorithms \texttt{BasicBB} and \texttt{FastBB}. Section~\ref{sec:framework} presents the frameworks \texttt{PB}, \texttt{IE} and \texttt{PBIE}. We conduct extensive experiments in Section~\ref{sec:exp}. Section~\ref{sec:related} reviews the related work and Section~\ref{sec:conclusion} concludes the paper.

\section{problem definition}
\label{sec:problem}
Let $G=(L\cup R,E)$ be an undirected and unweighted bipartite graph, where $L$ and $R$ are two disjoint vertex sets and $E$ is an edge set. 
For the graph $G$, we use $V(G)$, $L(G)$, $R(G)$, and $E(G)$ to denote its set of vertices, left side, right side and set of edges, respectively, i.e., $V(G)=L\cup R$, $L(G)=L$, $R(G)=R$, and $E(G) = E$. 
Given $X\subseteq L$ and $Y\subseteq R$, we use $G[X\cup Y]$ to denote the induced (bipartite) graph of $G$, i.e., $G[X\cup Y]$ includes the set of vertices $X\cup Y$ and the set of edges between $X$ and $Y$.
All subgraphs considered in this paper are induced subgraphs. 
We use $H$ or $(X,Y)$ as a shorthand of $H=G[X\cup Y]$.

Given $v\in L$, we use $\Gamma(v,R)$ (resp. $\overline{\Gamma}(v,R)$) to denote the set of neighbours (resp. non-neighbours) of $v$ in $R$, i.e., $\Gamma(v,R)=\{u\mid (v,u)\in E\ \text{and}\ u\in R\}$ (resp. $\overline{\Gamma}(v,R)=\{u\mid (v,u)\notin E\ \text{and}\ u\in R\}$). We define $\delta(v,R)=|\Gamma(v,R)|$ and $\overline{\delta}(v,R)=|\overline{\Gamma}(v,R)|$. 
We use $d$ to denote the maximum degree of vertex in $G$.
We have symmetric definitions for each vertex $u\in R$. 
%
Next, we review the definition of $k$-biplex~\cite{yu2021efficient}.

\begin{definition}[$k$-biplex~\cite{yu2021efficient}]
Given a graph $G = (L, R, E)$, a positive integer $k$, $X\subseteq L$ and $Y\subseteq R$, a subgraph $G[X\cup Y]$ is said to be a $k$-biplex
if $\overline{\delta}(v,Y)\leq k$, $\forall v\in X$ and $\overline{\delta}(u,X)\leq k$, $\forall u\in Y$.
\end{definition}

A $k$-biplex $H$ is said to be maximal if there is no other $k$-biplex $H'$ containing $H$, i.e., $V(H)\subseteq V(H')$. 
{Large real graphs usually involve numerous maximal $k$-biplexes and most of them highly overlap.} 
In this paper, we aim to find $K$ maximal $k$-biplexes with the most edges, where $K$ is a positive integral user-parameter.
In addition, we consider two size constraints $\theta_L$ and $\theta_R$ on each maximal $k$-biplex $H$ to be found, namely $|L(H)|\geq \theta_L$ and $|R(H)|\geq \theta_R$. These constraints would help to filter out some skewed maximal $k$-biplexes, i.e., the number of vertices at one side is extremely larger than that at the other side. 
{To guarantee that all found maximal $k$-biplexes are connected, we further require $\theta_L\geq 2k+1$ and $\theta_R\geq 2k+1$ based on the following lemma.}

\begin{lemma}
A $k$-biplex $H$ is connected if $|L(H)|\geq 2k+1$ and $|R(H)|\geq 2k+1$.
\end{lemma}
\begin{proof}
\roundA
This can be proved by contradiction. Suppose $H$ is not connected and is partitioned into two connected components, namely $H_1$ and $H_2$. We derive the contradiction by showing that $H$ is not a $k$-biplex: for a vertex $v_1$ in $L(H_1)$, it disconnects more than $k$ vertices, i.e., $\overline{\delta}(v_1,R(H))\geq |R(H_2)| \geq k+1$. Specifically, we derive $\overline{\delta}(v_1,R(H))\geq |R(H_2)|$ since $v_1$ from $H_1$ disconnects all vertices in $H_2$ based on the assumption, and we derive $|R(H_2)|\geq k+1$ by (1) for a vertex $v_2$ in $L_2$, $|R(H_2)|\geq \delta(v_2,R(H))$ since all neighbours of $v_2$ within $H$ reside in $R(H_2)$ based on the assumption and (2) $\delta(v_2,R(H))\geq R(H)-k\geq 2k+1-k\geq k+1$ since $v_2$ disconnects at most $k$ vertices ($H$ is a $k$-biplex) and $R(H)\geq 2k+1$.  
\end{proof}

We formalize the problem studied in this paper as follows. 
\begin{problem}[Maximum $k$-biplex Search] 
Given a bipartite graph $G=(L\cup R,E)$, four positive integers $K > 0$, $k> 0$, $\theta_L\geq 2k+1$ and $\theta_R\geq 2k+1$, the maximum $k$-biplex search problem aims to find $K$ maximal $k$-biplexes such that each found maximal $k$-biplex $H$ satisfies that $|L(H)|\geq \theta_L$, $|R(H)|\geq \theta_R$ and $|E(H)|$ is larger than $|E(H')|$ for any other maximal $k$-biplex $H'$ that is not returned.
\end{problem}

In this paper, we use MBP and MaxBP as a shorthand of a maximal $k$-biplex and one of the $K$ maximal $k$-biplexes with the most edges, respectively.
%

{
\smallskip
\noindent\textbf{NP-hardness.} The maximum $k$-biplex search problem is NP-hard, which we present in the following lemma (with its proof provided 
{\LC in Section~\ref{sec:proof}}
).
\begin{lemma}
The maximum $k$-biplex search problem is NP-hard.
\label{lemma:np}
\end{lemma}
}

\noindent\textbf{Remarks.}
In the following sections (Section~\ref{sec:alg} and Section~\ref{sec:framework}), we focus on the setting of $K = 1$ (i.e., the problem becomes to find a MBP with the maximum number of edges) when presenting the algorithms for ease of presentation. We note that these algorithms can be naturally extended for general settings of $K$ with minimal efforts (i.e., we maintain $K$ MaxBPs instead of $1$ MaxBP found so far throughout the algorithm for pruning) and are tested in Section~\ref{sec:exp}.

\section{Branch-and-Bound Algorithms}
\label{sec:alg}

 We first introduce a \emph{branch-and-bound} algorithm called \texttt{BasicBB}, which is based on a conventional and widely-used branching strategy that we call \emph{Bron-Kerbosch (BK) branching}~\cite{bron1973algorithm}, in Section~\ref{subsec:basicbb}. \texttt{BasicBB} has the worst-case time complexity $O(|V|\cdot d\cdot 2^{|V|})$. 
 We then introduce a new branching strategy called \emph{Symmetric-BK (Sym-BK) branching}, which is symmetric to the BK branching but better suits our problem of finding MaxBP, in Section~\ref{subsec:sym-bk}. We further present our method for determining an ordering of vertices, which is critical for Sym-BK branching, in Section~\ref{subsec:sym-bk-ordering}. We finally introduce a new branch-and-bound algorithm called \texttt{FastBB}, which is based on Sym-BK branching, and analyze its time complexity in Section~\ref{subsec:fastbb}. \texttt{FastBB} has its worst-case time complexity $O(|V|\cdot d \cdot \gamma_k^{|V|})$, where $\gamma_k$ is strictly smaller than 2 and depends on the setting $k$. For example, when $k=1$, $\gamma_k$ is $1.754$.

\begin{figure}[t]
	\centering
	\vspace{-0.1in}
	\includegraphics[width=0.85\linewidth]{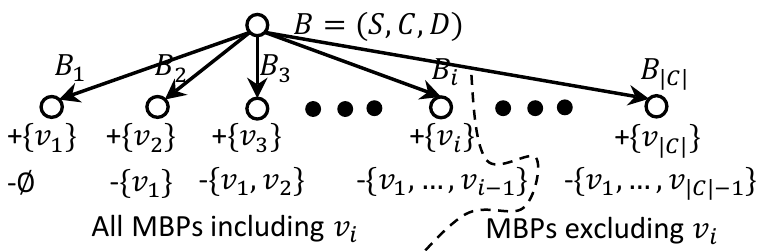}
	\\
	(a) BK Branching
	\vspace{0.10in}
	\\
	\includegraphics[width=0.85\linewidth]{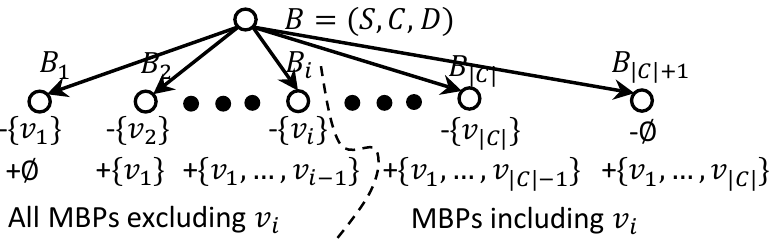}
	\\
	(b) Sym-BK Branching
	\vspace{-0.15in}
	\caption{Illustration of two branching strategies. Each node denotes a branch $(S,C,D)$. The notation ``+'' means to include a vertex by adding it $S$ and ``-'' means to exclude a vertex by adding it to $D$.}
	\label{fig:branching_scheme}
\end{figure}


\subsection{A BK Branching based Branch-and-Bound Algorithm: \texttt{BasicBB}}
\label{subsec:basicbb}



Our first attempt is to adapt the seminal \emph{Bron-Kerbosch} (BK) algorithm~\cite{bron1973algorithm}. It recursively partitions the search space (i.e., the set of all possible MBPs) to multiple sub-spaces via \emph{branching}.
Specifically, each sub-space is represented by a triplet of three sets $(S, C, D)$ as explained below.
\begin{itemize}[leftmargin=*]
	\item \textbf{Partial set $S$.} Set of the vertices that \emph{must} be included in every MBP within the space.
	\item \textbf{Candidate set $C$.} Set of the vertices that \emph{can} be included in $S$ in order to form a MBP within the space.
	\item \textbf{Exclusion set $D$.} Set of vertices that \emph{must not} be included in any MBP within the space.
\end{itemize}
We further denote by $S_L$, $C_L$, and $D_L$ the left side of $S$, $C$, and $D$, respectively, and define $S_R$, $C_R$, and $D_R$ similarly for the right side.

The recursive process of the BK algorithm starts from the full search space with $S = \emptyset$, $C = V$, and $D = \emptyset$. 
Consider the branching step at a current branch $B = (S, C, D)$.
Let $\langle v_1, v_2, ..., v_{|C|} \rangle$ be a sequence of the vertices in $C$.
The branching step would partition the space to $|C|$ sub-spaces (and correspondingly $|C|$ branches), where the $i^{th}$ branch, denoted by $B_i = (S_i, C_i, D_i)$, \emph{includes} $S$ and $v_i$ and \emph{excludes} $v_1, v_2, ..., v_{i-1}$. Formally, for $1\le i \le |C|$, we have 
\begin{equation}
S_i = S \cup \{ v_i\};~~D_i = D \cup \{v_1, v_2, ..., v_{i-1}\};~~C_i = C - \{ v_1, v_2, ..., v_i\}
\label{equation:bk}
\end{equation}
For illustration, consider Figure~\ref{fig:branching_scheme}(a). 
Note that $C_i$ and $D_i$ can be further refined by removing those vertices that cannot be included to $G[S_i]$ to form a $k$-biplex.
%

We call the above branching strategy \emph{BK branching}. BK branching essentially corresponds to a recursive \emph{binary} branching process. It first splits the current branch into two branches, one including $v_1$ (this is the branch $B_1$) and the other excluding $v_1$. Then, it further splits the latter into two branches, one including $v_2$ (this is the branch $B_2$) and the other excluding $v_2$. It continues the process until the last branch, which excludes $v_1, v_2, ..., v_{|C|-1}$ and includes $v_{|C|}$ (this corresponds to the branch $B_{|C|}$), is formed. In particular, the branches $B_1, B_2, ..., B_i$ cover all MBPs including $v_i$ and branches $B_{i+1}, ..., B_{|C|}$ cover those excluding $v_i$, as indicated by the dashed line in Figure~\ref{fig:branching_scheme}(a).

%
We note that BK branching relies on an ordering of vertices in the candidate set $C$, i.e., $\langle v_1, v_2, ..., v_{|C|} \rangle$, for producing branches. {\roundA In this paper, we follow the existing studies~\cite{zhang2014finding,DBLP:conf/ijcai/AbidiZCL20} and use the \emph{non-decreasing vertex degree ordering} (where vertices are ranked in a non-decreasing order of their degrees in $S\cup C$, i.e., $\delta(v_i,S\cup C)\leq \delta(v_j, S\cup C)$ for any $i<j$) since this would help with effective pruning as shown empirically.} We note that normally the ordering does not affect the worst-case theoretical time complexity of the algorithm based on BK branching.

During the recursive search process, some pruning techniques can be applied. 
Let $B = (S, C, D)$ be a branch. \underline{First}, branch $B$ can be pruned if $S$ is not a $k$-biplex since (1) each partial set in the search space corresponding to this branch is a \emph{superset} of $S$ and (2) based on the hereditary property of $k$-biplex, any superset of a non-$k$-biplex is not a $k$-biplex. 
\underline{Second}, branch $B$ can be pruned if an upper bound of the left side (resp. the right side) of a $k$-biplex in the space is smaller than $\theta_L$ (resp. $\theta_R$) based on the problem definition.
\underline{Third}, branch $B$ can be pruned if an upper bound of the number of edges in a $k$-biplex in the space is smaller than the largest one of a $k$-biplex known so far. {\revision\underline{Fourth}}, branch $B$ can be pruned if there exists a vertex in $D$ such that including this vertex to each $k$-biplex in the space would still result in a $k$-biplex. 
We will elaborate on these pruning rules in detail in Section~\ref{subsec:fastbb}.
%
Finally, the recursive process of the BK algorithm terminates at a branch $B = (S, C, D)$ if $G[S\cup C]$ is a $k$-biplex since $G[S \cup C]$ would be the MaxBP within the space of the branch.


We call this BK algorithm, which is  a \emph{branch-and-bound} algorithm based on BK branching and the aforementioned four pruning techniques, \texttt{BasicBB}, and present its pseudo-code in Algorithm~\ref{alg:enumeration_scheme}. 
%
Similar to many existing algorithms that are based on BK branching, the worst-case time complexity of \texttt{BasicBB} is $O(|V|\cdot d\cdot 2^{|V|})$ (i.e., $O^*(2^{|V|})$)~\cite{yu2021efficient,DBLP:journals/sadm/SimLGL09}, 
though its practical performance can be boosted by the the pruning techniques.

\begin{algorithm}{}
\small
\caption{The branch-and-bound algorithm based on BK branching: \texttt{BasicBB}}
\label{alg:enumeration_scheme}
\KwIn{A graph $G(L\cup R,E)$, $k$, $\theta_L$ and $\theta_R$}
\KwOut{The maximal $k$-biplex $H^*$ with the most edges}
$H^* \leftarrow G[\emptyset]$; \tcp{Global variable}
\texttt{BasicBB-Rec}$(\emptyset,L\cup R,\emptyset)$;\ \ \textbf{return} $H^*$\;
\SetKwBlock{Enum}{Procedure \texttt{BasicBB-Rec}$(S,C,D)$}{}
\SetKwBlock{update}{Procedure \texttt{Update}$(S,C)$}{}
\Enum{

    \tcc{Termination}
    \If{$G[S \cup C]$ is $k$-biplex}{
       {$H^*\leftarrow G[S \cup C]$ if $|E(G[S \cup C])|> |E(H^*)|$};\ \textbf{return}; 
    }
    
    \tcc{Pruning}
    \lIf{any of pruning conditions is satisfied (details in Section~\ref{subsec:fastbb})}{
        \textbf{return} 
    }
   
    \tcc{BK Branching}
    Create $|C|$ branches $B_i = (S_i, C_i, D_i)$ based on Equation~(\ref{equation:bk});\\
    \For{each branch $B_i$}{
            {\revision \texttt{BasicBB-Rec}$(S_i,C_i,D_i)$;}
    }
}
\end{algorithm}

\subsection{A New Branching Strategy: Sym-BK Branching}
\label{subsec:sym-bk}

We observe that there exists a branching strategy, which is natural and \emph{symmetric} to BK branching. 
Specifically, consider the branching step at a current branch $B = (S, C, D)$.
Let $\langle v_1, v_2, ..., v_{|C|} \rangle$ be a sequence of the vertices in $C$. This branching step would partition the space to $(|C| + 1)$ sub-spaces (and correspondingly $(|C| + 1)$ branches), where the $i^{th}$ branch, denoted by $B_i = (S_i, C_i, D_i)$, \emph{includes} $S$ and $v_1, v_2, ..., v_{i-1}$ and \emph{excludes} $v_i$. Here, $v_0$ and $v_{|C|+1}$ both correspond to null. Formally, for $1\le i \le |C|+1$, we have 
\begin{equation}
S_i = S \cup \{ v_1, v_2, ..., v_{i-1}\};~~D_i = D \cup \{v_i\};~~ C_i = C - \{ v_1, v_2, ..., v_i\}
\label{equation:sym-bk}
\end{equation}
For illustration, consider Figure~\ref{fig:branching_scheme}(b).  Note that $C_i$ and $D_i$ can be further refined by removing those vertices that cannot be included to $G[S_i]$ to form a $k$-biplex.

We call the above branching strategy \emph{symmetric-BK (Sym-BK) branching}. Sym-BK branching corresponds to another recursive \emph{binary} branching process, which is symmetric to that of the BK branching. Specifically, 
it first splits the current branch into two branches, one excluding $v_1$ (this is the branch $B_1$) and the other including $v_1$. Then, it further splits the latter into two branches, one excluding $v_2$ (this is the branch $B_2$) and the other including $v_2$. It continues the process until the last branch, which includes $v_1, v_2, ..., v_{|C|}$ (this corresponds to the branch $B_{|C|+1}$), is formed. In particular, the branches $B_1, B_2, ..., B_i$ cover all MBPs \emph{excluding} $v_i$ and branches $B_{i+1}, ..., B_{|C|+1}$ cover those \emph{including} $v_i$, as indicated by the dashed line in Figure~\ref{fig:branching_scheme}(b).

%
%
%



\smallskip
\noindent\textbf{Sym-BK branching vs. BK branching.} They are symmetric to each other and both of them are natural branching strategies. 
They differ in that among two branches formed by a binary branching based on a vertex, BK branching recursively partitions the branch \emph{excluding} the vertex while Sym-BK recursively partitions one \emph{including} the vertex.
Sym-BK branching produces one more branch than BK branching at each branching step (i.e., $|C|+1$ branches vs. $|C|$ branches), but this difference of one extra branch is negligible given that there can be many branches produced at the branching step. 
Compared with BK branching, Sym-BK branching has the following advantages when adopted for our problem of finding the MaxBP.

\underline{First}, at each branching step, it would produce branches with \emph{bigger} partial sets $S_i$ (note that the $i^{th}$ branch by Sym-BK branching involves $|S| + (i-1)$ vertices in the partial set while that by BK branching involves $|S| + 1$ vertices). Consequently, the produced branch would have a larger chance to be pruned due to the hereditary property of $k$-biplex (if a set of vertices is not a $k$-biplex, then none of its supersets is, but not vice versa). 
\underline{Second}, the partial set of the $j^{th}$ branch, i.e., $S_j$, is always a superset of that of the $i^{th}$ branch, i.e., $S_i$, for any $j > i$. Consequently, if $S_i$ is not a $k$-biplex (which means the branch $B_i$ can be pruned), then all branches following $B_i$ can be pruned (since their partial sets are supersets of $S_i$ and thus they are not $k$-biplexes either based on the hereditary property). 

To illustrate, consider the example in Figure~\ref{fig:branching_case}. 
One branching step of Sym-BK branching is shown in Figure~\ref{fig:branching_case}(b). The fourth branch has the partial set of $S_4 = \{u_0,v_0,u_1,u_2,u_4\}$, which is not a $k$-biplex. All the following branches have their partial sets as supersets of $S_4$, and thus they can be pruned immediately (as indicated by the shaded color in the figure). 

\begin{figure}[t]
	\centering
	\vspace{-0.1in}
	\includegraphics[width=0.78\linewidth]{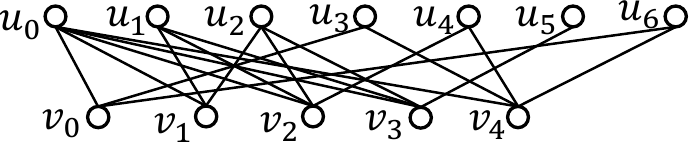}
	\\
	(a) Input graph used throughout the paper
	\vspace{0.10in}
	\\
	\includegraphics[width=0.85\linewidth]{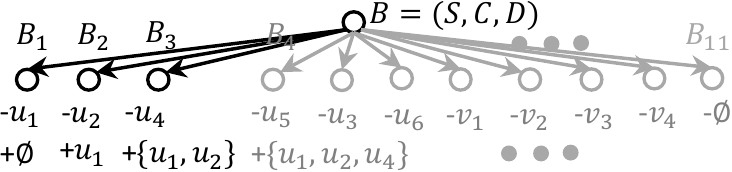}
	\\
	(b) Case 1 ($S=\{u_0,v_0\}$, $C=\{u_1,u_2,u_4,u_5,u_3,u_6,v_1,v_2,v_3,v_4\}$ and $D=\emptyset$)
	\vspace{0.10in}
	\\
	\includegraphics[width=0.85\linewidth]{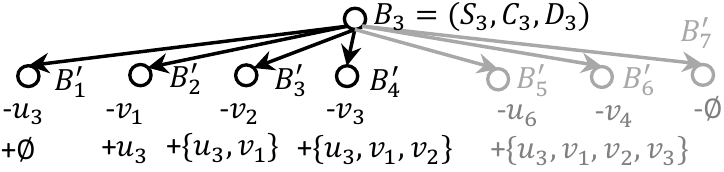}
	\\
	(c) Case 2 ($S_3\!=\!\{u_0,v_0,u_1,u_2\}$, $C_3\!=\!\{u_3,v_1,v_2,v_3,u_6,v_4\}$ and $D_3\!=\!\{u_4,u_5\}$)
 	\vspace{-0.15in}
	\caption{Illustration of Sym-BK branching ($k=2$).}
	\label{fig:branching_case}
	\vspace{-0.15in}
\end{figure}

In fact, with Sym-BK branching and a carefully-designed ordering of vertices (details will be introduced in Section~\ref{subsec:sym-bk-ordering}), our new branch-and-bound algorithm \texttt{FastBB} would have the worst-case time complexity of $O(|V|\cdot d\cdot \gamma_k^{|V|})$ with $\gamma_k < 2$, which is strictly smaller than that of the \texttt{BasicBB} algorithm based on BK branching (details will be introduced in Section~\ref{subsec:fastbb}).

\smallskip
\noindent\textbf{Remarks.}
In~\cite{DBLP:conf/aaai/ZhouXGXJ20}, the authors design a branching strategy, which performs one of two branching operations at a branch depending on the situation: (1) generating two branches based on a \emph{single} vertex in $C$ (i.e., one including and the other not including this vertex) and (2) generating at most $(|C|+1)$ branches based on all vertices in $C$ with pruning. We call this branching strategy \emph{hybrid branching}. Our Sym-BK branching is superior over hybrid branching for our problem in two aspects. First, Sym-BK branching has a more simplified form. Second, the branch-and-bound algorithm based on Sym-BK branching has lower worst-case time complexity \emph{theoretically} (details are in Section~\ref{subsec:fastbb}) and runs faster \emph{empirically} (details are in Section~\ref{sec:exp}).

\subsection{Sym-BK Branching: Ordering of Vertices}
\label{subsec:sym-bk-ordering}

The Sym-BK branching relies on an ordering of the vertices in $C$.
Recall that in the branches $B_1, B_2, ..., B_{|C|+1}$ generated by Sym-BK branching, the partial set of a branch $B_j$ is always a superset of that of a preceding branch $B_i$ ($i < j$). Therefore, our idea is to figure out a small subset $C'$ of vertices in $C$ such that including them \emph{collectively} to the partial set $S$ would violate the $k$-biplex definition. We then put these vertices \emph{before} other vertices in the ordering. 
In this way, the branch with the partial set of $S \cup C'$ and all the following branches can be pruned directly. 
We elaborate on this idea in detail next.

We notice that $G[S\cup C]$ is not a $k$-biplex since otherwise the recursion would terminate at this branch. It means that there exists at least a vertex in $S \cup C$, which has more than $k$ disconnections within $G[S\cup C]$. Without loss of generality, we assume that the vertex is from the left side and denote it by $\hat{v} \in S_L \cup C_L$. Consider the set of vertices that disconnect $\hat{v}$ in $C_R$, i.e., $\overline{ \Gamma} ( \hat{v}, C_R)$. We know that including $\overline{ \Gamma} ( \hat{v}, C_R)$  to $S$ collectively would violate the $k$-biplex definition. Specifically, if $\hat{v}$ is already in $S_L$, i.e., $\hat{v} \in S_L$, we can include at most $k-\overline{\delta}(\hat{v}, S_R)$ vertices from $\overline{ \Gamma} ( \hat{v}, C_R)$ to $S$ without violating the $k$-biplex definition; and if $\hat{v}$ is not yet in $S_L$, i.e., $\hat{v} \in C_L$, we can include $\hat{v}$ together with at most $k-\overline{\delta}(\hat{v}, S_R)$ vertices from $\overline{ \Gamma} ( \hat{v}, C_R)$ to $S$ without violating the $k$-biplex definition. For the simplicity of notations, we define
\begin{equation}
a = k-\overline{\delta}(\hat{v}, S_R); ~~b = \overline{\delta}(\hat{v}, C_R).
\end{equation}
Intuitively, $a$ means the greatest possible number of disconnections that $\hat{v}$ can have when including more vertices from $C$ to $S$ for forming MBPs.
Note that we have $0\le a \le k$ (since $S$ is a $k$-biplex and thus $\overline{\delta}(\hat{v}, S_R) \le k$) and $a < b$ (since $b - a = \overline{\delta}(\hat{v}, S_R \cup C_R) - k > 0$). 
Based on $\hat{v}$, we define an ordering of the vertices in $C$.

\smallskip
\noindent\textbf{Case 1}: $\hat{v} \in S_L$. In this case, we define the ordering as follows.
\begin{equation}
	\langle u_1, u_2, ..., u_b, u_{b+1}, ..., u_{|C|} \rangle,
	\label{equation:sym-bk-case1}
\end{equation}
where $u_1, u_2, ..., u_b$ are vertices from $\overline{\Gamma}(\hat{v}, C_R)$ in any order and $u_{b+1}, u_{b+2}, ..., u_{|C|}$ are vertices from $C  - \overline{\Gamma}(\hat{v}, C_R)$ in any order.
Based on this ordering, the branch $B_{a+2}$ would have the partial set $S\cup \{u_1, u_2, ..., u_{a+1}\}$, which is not a $k$-biplex (since $\hat{v}$ would have more than $k$ disconnections). Therefore, branches $B_{a+2}, B_{a+3}, ..., B_{|C|+1}$ can be pruned and only the first $(a+1)$ branches, namely $B_1, B_2, ..., B_{a+1}$, would be kept.
{\roundA
To illustrate, we consider a branch with $S=\{u_0,v_0\}$, $D=\emptyset$ and $C=\{u_1,u_2,u_3,u_4,u_5,u_6,v_1,v_2,v_3,v_4\}$ for finding a MaxBP with $k=2$ from the input graph in Figure~\ref{fig:branching_case}(a). Based on $v_0$ in $S$ that disconnects 3 vertices, i.e., $\{u_1,u_2,u_4\}$, we define the ordering $\langle u_1,u_2,u_4,u_5,u_3,u_6,v_1,v_2,v_3,v_4 \rangle$ for Sym-BK branching as shown in Figure~\ref{fig:branching_case}(b). The branches $B_4, ..., B_{11}$ can be pruned since $B_4$ has the partial set $\{u_0,u_1,u_2,u_4,v_0\}$ not a $k$-biplex ($v_0$ has more than $k=2$ disconnections).
}

\smallskip
\noindent\textbf{Case 2}: $\hat{v} \in C_L$. In this case, we define the ordering as follows.
\begin{equation}
	\langle \hat{v}, u_1, u_2, ..., u_b, u_{b+2}, ..., u_{|C|} \rangle,
	\label{equation:sym-bk-case2}
\end{equation}
where $u_1, u_2, ..., u_b$ are vertices from $\overline{\Gamma}(\hat{v}, C_R)$ in any order and $u_{b+2}, u_{b+2}, ..., u_{|C|}$ are vertices from $C  - \overline{\Gamma}(\hat{v}, C_R) - \{\hat{v}\}$ in any order.
Based on this ordering, the branch $B_{a+3}$ would have the partial set as $S\cup \{\hat{v}, u_1, u_2, ..., u_{a+1}\}$, which is not a $k$-biplex (since $\hat{v}$ would have more than $k$ disconnections). Therefore, branches $B_{a+3}, B_{a+4}, ..., B_{|C|+1}$ can be pruned and only the first $(a+2)$ branches, namely $B_1, B_2, ..., B_{a+2}$, would be kept.
{\roundA
To illustrate, we consider another branch $B_3$ with $S_3=\{u_0,u_1,u_2,v_0\}$, $C_3=\{u_3,u_6,v_1,v_2,v_3,v_4\}$ and $D_3=\{u_4,u_5\}$ in Figure~\ref{fig:branching_case}(c). Based on $u_3$ in $C$ that disconnects 3 vertices, i.e., $\{v_1,v_2,v_3\}$, we define the ordering $\langle u_3,v_1,v_2,v_3,u_6,v_4 \rangle$ for Sym-BK branching as shown in Figure~\ref{fig:branching_case}(c). 
The branches $B_5'$, $B_6'$ and $B_7'$ can be pruned since $B_5'$ has the partial set as $\{u_0,u_1,u_2,u_3,v_0,v_1,v_2,v_3\}$ not a $k$-biplex ($u_3$ has more than $k=2$ disconnections).
}

We note that there could be multiple vertices, which have more than $k$ disconnections among $S\cup R$, and for each of them, we can define an ordering as above. We call these vertices \emph{candidate pivots} and the vertex that we pick for defining an ordering the \emph{pivot}. An immediate question is: \emph{which one should we select as the pivot among the candidate pivots?} To answer this question, we quantify the benefits of specifying the ordering based on a specific candidate pivot $\hat{v}$. There are two benefits (for simplicity, we discuss the case of $\hat{v} \in S_L$ only, and the other case is similar and thus omitted). \underline{Benefit 1:} $(|C| - a)$ branches, namely $B_{a+2}$, $B_{a+3}$, ..., $B_{|C|+1}$, are pruned for $\hat{v}$. 
Therefore, the smaller $a$ is, the larger the Benefit 1 is.
\underline{Benefit 2:} For Branch $B_{a+1}$, we have $|C_{a+1}| \le |C| - b$ since (1) $C_{a+1}$ is updated to be $C - \{u_1, u_2, ..., u_{a+1}\}$ (please refer to Equation~(\ref{equation:sym-bk-case1})) and (2) the vertices $u_{a+2}, u_{a+3}, ..., u_b$ can be further excluded from $C_{a+1}$ since including each of these vertices to $S_{a+1}$ would violate the $k$-biplex definition. 
Therefore, the larger $b$ is, the larger the Benefit 2 is.
In summary, for a vertex with a \emph{smaller} $a$ and/or a \emph{larger} $b$, the overall benefits would be more significant. Therefore, we select the candidate pivot $\hat{v}$ with the largest $(b - a) = \overline{\delta}(\hat{v}, S_R \cup C_R) - k$ as the pivot. Equivalently, it would select the candidate pivot with the most disconnections within $S \cup C$. Furthermore, to achieve a better worst-case time complexity (details will be introduced in Section~\ref{subsec:fastbb}), we first select the pivot among the pivot candidates in $S$ if possible; otherwise, we select one among those in $C$.

{\roundA
To illustrate, we consider again the example in Figure~\ref{fig:branching_case}. For a branch $B$ in Figure~\ref{fig:branching_case}(b), $v_0$ would be selected as the pivot since $v_0$ in $S$ has the number of disconnections more than $k$ and the greatest among other vertices in $S$. For another branch $B_3$ in Figure~\ref{fig:branching_case}(c), $u_3$ would be selected as the pivot since (1) every vertex in $S$ disconnects less than $k$ vertices and (2) $u_3$ has the number of disconnections more than $k$ and the greatest among other vertices in $C$.
}

{\revision
\smallskip
\noindent\textbf{Remarks.} We remark that the branch-and-bound method in~\cite{DBLP:conf/aaai/ZhouXGXJ20} selects the pivot with the smallest degree among all candidate pivots in $S\cup C$, from which our strategy differs in two aspects.
First, we select the pivot with the most disconnections as discussed above. Note that the pivot with the most disconnections could be different from that with the smallest degree in bipartite graphs. Second, we prioritize $S$ for selecting a pivot. The intuition behind is that (1) the worst-case running time of Case 1 (i.e., branching based on a pivot from $S$) is strictly smaller than that of Case 2 and (2) it would help to further improve the time complexity of Case 2 (details are in the proof of Theorem~\ref{lemma:time_complexity} in Section~\ref{subsec:fastbb}).}

\subsection{A Sym-BK Branching based Branch-and-Bound Algorithm: \texttt{FastBB}}
\label{subsec:fastbb}

Based on the Sym-BK branching strategy and the aforementioned pruning techniques (details will be presented in this section), we design a branch-and-bound algorithm, called \texttt{FastBB}. The pseudo-code of \texttt{FastBB} is presented in Algorithm~\ref{alg:maxbp}.
\texttt{FastBB} differs from \texttt{BasicBB} only in the branching step (i.e., Lines 7 - 10 of Algorithm~\ref{alg:maxbp}). 
Next, we elaborate on the pruning conditions (Line 6 in Algorithm~\ref{alg:maxbp}) in detail.

%
%

\noindent\textbf{Pruning conditions.} 
Let $B = (S, C, D)$ be a branch.
We first define $\tau_L = \min_{u\in S_R}\delta(u,S_L\cup C_L) + k$ and $\tau_R = \min_{v\in S_L}\delta(v,S_R\cup C_R) + k$,
which can be verified to be the upper bound of the number of vertices at the left side and that at the right side of a MBP covered by the branch $B$, respectively.
%
We can prune the branch $B$ if any of the following four conditions is satisfied.
\begin{enumerate}[leftmargin=*]
    \item $S$ is not a $k$-biplex.
    \item $\tau_L<\theta_L$ or $\tau_R<\theta_R$. 
    \item $|E(G[S\cup C])|\leq |E(H^*)|$ or $\tau_L\times \tau_R \leq |E(H^*)|$, where $H^*$ is the MaxBP found so far. 
    \item There exists a vertex $v\in D_L$ such that $\overline{\delta}(v,S_R\cup C_R)\leq k$ and $\{w\in S_R\cup C_R \mid \overline{\delta}(w,S_L\cup C_L)\geq k\}\subseteq \Gamma(v,R)$ or symmetrically there exists such a vertex $u\in D_R$. 
\end{enumerate}
Condition (1) holds because of the hereditary property of $k$-biplex, Condition (2) is based on the size constraints of the two sides of MaxBP to be found, Condition (3) is based on the objective of the problem (i.e., to maximize the number of edges in a MBP), and Condition (4) holds because all $k$-biplexes covered by this branch (if any) would not be maximal (since an additional vertex $v$ or $u$ can be included in each of them without violating the $k$-biplex definition).

\begin{algorithm}{}
\small
\caption{The branch-and-bound algorithm based on Sym-BK branching: \texttt{FastBB}}
\label{alg:maxbp}
\KwIn{A graph $G(L\cup R,E)$, $k$, $\theta_L$ and $\theta_R$}
\KwOut{The maximal $k$-biplex $H^*$ with the most edges}
$H^* \leftarrow G[\emptyset]$; \tcp{Global variable}
\texttt{FastBB-Rec}$(\emptyset,L\cup R,\emptyset)$;\ \ \textbf{return} $H^*$\;

\SetKwBlock{Enum}{Procedure \texttt{FastBB-Rec}$(S,C,D)$}{}
\Enum{
    \tcc{Termination}
    \If{$G[S\cup C]$ is a $k$-biplex}{
         $H^*\leftarrow G[S\cup C]$ if $|E(G[S\cup C])|>|E(H^*)|$ and \textbf{return} 
    }
    
    \tcc{Pruning}
    \lIf{any of pruning conditions is satisfied
    }{
        \textbf{return} 
    }
    
    \tcc{Sym-BK Branching}
    Select a pivot vertex $\hat{v}$ and determine an ordering based on $\hat{v}$ (Section~\ref{subsec:sym-bk-ordering});\\
    \lIf{$\hat{v}\in S$}{
        Create $a+1$ branches $\{B_1,B_2,...,B_{a+1}\}$ (Equation~(\ref{equation:sym-bk}) and (\ref{equation:sym-bk-case1}))
    }\ElseIf{$\hat{v}\in C$}{
        Create $a+2$ branches $\{B_1,B_2,...,B_{a+2}\}$ (Equation~(\ref{equation:sym-bk}) and (\ref{equation:sym-bk-case2}))\;
    }
    \For{each {\LC created branch $B_i$}
    }{
        \texttt{FastBB-Rec}$(S_i,C_i,D_i)$
    }
}
\end{algorithm}

\smallskip
\noindent\textbf{Worst-case time complexity}. 
The worst-case time complexity of \texttt{FastBB} is strictly better than than of \texttt{BasicBB}, which we show in the following theorem.

\begin{theorem}
\label{lemma:time_complexity}
Given a bipartite graph $G$, \texttt{FastBB} finds the MaxBP in time $O(|V|\cdot d\cdot \gamma_k^{|V|})$ where $\gamma_k$ is the largest positive real root of $x^{k+4}-2x^{k+3}+x^2-x+1=0$. For example, when $k=1$, $2$ and $3$, $\gamma_k=1.754$, $1.888$ and $1.947$, respectively.
\end{theorem}

{\roundA
\begin{proof}
We give a sketch of the proof and put the details
\ifx \CR\undefined
{\LC in the appendix.}
\else
{\LC in the technical report~\cite{TR}.}
\fi
We recursively maintain two arrays to record the degree of each vertex $v$ within $G[S]$ or $G[S\cup C]$, i.e., $\delta(v,S)$ or $\delta(v,S\cup C)$. Then, the recursion of \texttt{FastBB-Rec} runs in polynomial time $O(|V|\cdot d)$. Specifically, the time cost is dominated by the part of checking pruning condition (4) in line 6. This part has two steps, namely finding all those vertices with at least $k$ disconnections from $S\cup C$ in $O(|S\cup C|)$ time and checking the pruning condition (4) for each vertex in $D$ in $O(|D|\cdot d)$ time, where $|S\cup C|$ and $|D|$ are both bounded by $O(|V|)$.
%
%

Next, we analyze the number of recursions. Let $T(n)$ be the largest number of recursions where $n=|C|$. We have two cases.

\noindent\textbf{Case 1 ($\hat{v}\in S$).} We remove $i$ vertices from $C$ in $B_i$ ($1\leq i\leq a$) and $b$ vertices from $C$ in $B_{a+1}$ in the worst-case. Hence, we have
\begin{equation}
\label{eq:time1}
    T_1(n)\leq \sum_{i=1}^a T_1(n-i)+T_1(n-b).
\end{equation}
As discussed earlier, we have $a\leq k$ and $a<b$. It is easy to verify that we reach the maximum of $T_1(n)$ when $a=k$ and $b=k+1$. We thus have $T_1(n)\!\leq\! \sum_{i=1}^{k}T_1(n-i)\!+\!T_1(n-k-1)$. By solving this linear recurrence, the worst-case running time is $O(|v|\cdot d\cdot \gamma_k^{|V|})$ where $\gamma_k$ is the largest positive real root of $x^{k+2}-2x^{k+1}+1=0$. For example, $\gamma_k=1.618,1.839$ and $1.928$ when $k=1,2$ and 3, respectively.

\noindent\textbf{Case 2 ($\hat{v}\in C$).} We remove $i$ vertices from $C$ in $B_i$ ($1\leq i\leq a+1$) and $b+1$ vertices from $C$ in $B_{a+2}$ in the worst-case. Hence, we have
\begin{equation}
    T_2(n)\leq \sum_{i=1}^{a+1} T_2(n-i)+T_2(n-b-1).
\end{equation}
Assume $\hat{v}\in C_L$, we consider two scenarios, i.e.,  $\overline{\delta}(\hat{v},S_R\cup C_R)\geq k+2$ and $\overline{\delta}(\hat{v},S_R\cup C_R)=k+1$. 

\begin{itemize}[leftmargin=*]
    \item For Scenario 1, we can imply $a\leq k$ and $b>a+1$. When $a=k$ and $b=k+2$, $T_2(n)$ reaches the maximum. Hence, the worst-case running time is $O(|V|\cdot d\cdot \gamma_k^{|V|})$ where $\gamma_k$ is the largest positive real root of $x^{k+4}-2x^{k+3}+x^2-x+1=0$. For example, when $k=1$, 2 and 3, we have $\gamma_k=1.754$, 1.888 and 1.947, respectively.
    \item For Scenario 2, we can imply $a\leq k$ and $b>a$. When $a=k$ and $b=k+1$, $T_2(n)$ reaches the maximum. Thus the worst-case running time is $O(|V|\cdot d\cdot \gamma_k^{|V|})$ where $\gamma_k$ is the largest positive real root of $x^{k+3}-2x^{k+2}+1=0$. For example, when $k=1$, 2 and 3, we have $\gamma_k=1.839$, 1.928 and 1.966, respectively. 
\end{itemize}
The idea of remaining proof is to show that the analysis of Scenario 2 can be further improved based on our pivot selection strategy.
{\LC Specifically,}
Scenario 2 would have the worst-case time complexity smaller than {\LC that of} Scenario 1, and thus the worst-case time complexity of \texttt{FastBB} would be bounded by Scenario 1.
\end{proof}
}

{\revision
\noindent\textbf{Space complexity}. The space complexity of \texttt{FastBB} is $O(K\cdot (d+k)+k\cdot|V|)$, which is \emph{linear} to the size of the input bipartite graph given parameters $K$ and $k$. 
The first term $O(K\cdot (d+k))$ is the space cost of storing the returned top-$K$ MaxBPs (each has the size bounded by $2\times(d+k)$ since otherwise it involves at least $2\times(d+k)+1$ vertices and would not be a $k$-biplex due to the fact that (1) it has the size of one side at least $d+k+1$ and (2) there would be a vertex from the other side which disconnects at least $k+1$ vertices since $d$ is the maximum degree of vertex in $G$).
The second term $O(k\cdot|V|)$ is the space cost of the branch-and-bound search process.
Specifically, the space cost of the search process is dominated by that of storing the vertex orderings (Equation~(\ref{equation:sym-bk-case1}) or (\ref{equation:sym-bk-case2})) in all recursions. 
To reduce the space cost, instead of storing the ordering of all vertices in $C$, we store that of only those vertices that are before others in the ordering, i.e., the first $a+1$ (resp. $a+2$) vertices of Equation~(\ref{equation:sym-bk-case1}) (resp. (\ref{equation:sym-bk-case2})) depending on the cases.
This is because it creates the first $a+1$ or $a+2$ branches only. 
Considering $a\leq k$, we know that the space cost of storing the vertex ordering in one recursion is bounded by $O(k)$. 
Furthermore, it can be verified that the number of recursions is bounded by $O(|V|)$ (since whenever the number of recursions increases by one, at least one vertex is removed from $C$ and $C\subseteq V$). 
In summary, the overall space cost of storing the ordering of vertices is $O(k\cdot |V|)$.
%
}

{\revision
\smallskip
\noindent\textbf{Remark 1}. 
The baseline method \texttt{FPadp}, which adapts the \texttt{FaPlexen} method in~\cite{DBLP:conf/aaai/ZhouXGXJ20}, has the worst-case time complexity $O(|V|\cdot d\cdot \beta_k^{|V|})$,
where $\beta_k$ is the largest positive real root of $x^{k+3}-2x^{k+2}+1=0$.
We compare between \texttt{FPadp} and \texttt{FastBB}'s time complexities: (1) the polynomial in both complexities (i.e., the time cost for each recursion) is $O(|V|\cdot d)$ since 
both algorithms are branch-and-bound based
and (2) the exponential factors in \texttt{FastBB}'s time complexity are strictly smaller than those of \texttt{FPadp}, e.g., $\gamma_1=1.754$ vs. $\beta_1=1.839$, $\gamma_2=1.888$ vs. $\beta_2=1.928$ and $\gamma_3=1.947$ vs. $\beta_3=1.966$ for $k=1$, 2, and 3, respectively. 
The differences in the exponential factors are due to the differences in the strategies of branching and ordering vertices (i.e., \texttt{FastBB} uses the Sym-BK branching while \texttt{FPadp} uses a hybrid branching; \texttt{FastBB} prioritizes $S$ over $C$ when selecting a pivot for deciding the ordering of vertices while \texttt{FPadp} does not).
}


{\revision
\smallskip
\noindent\textbf{Remark 2}. We notice that the delay measure, i.e., the longest running time between any two consecutive outputs, is often used for evaluating enumeration algorithms~\cite{yu2022kbiplex,zhang2014finding}. In our problem setting, all solutions can only be confirmed and returned in one shot at the end due to its optimization nature. For such optimization-oriented settings~\cite{chen2021efficient,ignatov2018mixed,zhou2018towards,lyu2020maximum}, the delay measure is not suitable since it would be equal to the total running time trivially. 
}

\section{Efficiency and Scalability Boosting Techniques}
\label{sec:framework}

In this section, we further introduce three techniques for boosting the efficiency and scalability of the branch-and-bound (\texttt{BB}) algorithms introduced in Section~\ref{sec:alg}, namely \emph{progressive bounding} (\texttt{PB}) in Section~\ref{subsec:pb}, \emph{inclusion-exclusion} (\texttt{IE}) in Section~\ref{subsec:ie}, and \texttt{PBIE}, which combines \texttt{PB} and \texttt{IE}, in Section~\ref{subsec:pbie}. 
%


\subsection{Progressive Bounding Framework: \texttt{PB}}
\label{subsec:pb}
{\ChengSIGMOD{\texttt{PB} is adapted from an existing study of finding the biclique with the maximum number of edges~\cite{lyu2020maximum}.}}
The major idea of \texttt{PB} is to run a \texttt{BB} algorithm multiple times, and for each time, it imposes appropriate constraints on the MBP to be found, including lower and upper bounds of the number of vertices on both the left and right sides of the MBP. Then, it returns the MBP with the most edges found at different times. With the constraints captured by the lower and upper bounds, there are two benefits, namely (1) the \texttt{BB} algorithm can be run on a reduced graph instead of the original one and (2) the efficiency of the \texttt{BB} algorithm on the reduced graph can be further boosted with additional pruning techniques based on the upper bounds. 
Note that the pruning techniques are based on the lower bounds $\theta_L$ and $\theta_R$ only in the \texttt{BB} algorithms.

Let $H^*$ be the MaxBP with the maximum $E(H^*)$. We have the following prior knowledge about the number of vertices at the left and right sides of $H^*$, i.e., $|L(H^*)|$ and $|R(H^*)|$.
\begin{equation}
\theta_L \!\leq\! |L(H^*)| \!\leq\! \delta_{max}^R+k;~~~~
\theta_R \!\leq\! |R(H^*)| \!\leq\! \delta_{max}^L+k.
\end{equation}
where $\delta_{max}^R=\max_{u\in R}\!\delta(u, L)$ and $\delta_{max}^L=\max_{v\in L}\delta(v,R)$.
The lower bounds $\theta_L$ and $\theta_R$ are inherited from the problem definition. The upper bounds of $\delta_{max}^R+k$ and $\delta_{max}^L+k$ can be verified easily and the proofs for them are thus omitted.
We denote by $LB_L^i$, $UB_L^i$, $LB_R^i$, and $UB_R^i$ the lower bound and upper bound of the number of vertices at the left and right sides, respectively, which the \texttt{PB} would use to capture the constraints at the $i^{th}$ time. Then, \texttt{PB} would set these bounds \emph{progressively} as follows.
\begin{align}
LB_{L}^i &= \max \{LB_{L}^{i-1} / 2, \theta_L\};~~~~UB_{L}^i = LB_{L}^{i-1}; \label{equation:bounds-left}\\
LB_{R}^i &= \max \{ |E(H_{i-1}^*)| /UB_{L}^i, \theta_R \};~~~~UB_{R}^i = \delta_{max}^L+k;
\label{equation:bounds-right}
\end{align}
where $LB_{L}^0 = \delta_{max}^R+k$ and $H_{i-1}^*$ is the MBP found at the $(i-1)^{th}$ time and $H_0^*$ can be set as $G[\emptyset]$.
Essentially, it (1) splits the range of possible values of $|L(H^*)|$, namely $[\theta_L, \delta_{max}^R+k]$, into intervals with lengths decreasing \emph{logarithmically}, (2) uses the boundaries of the intervals as lower and upper bounds of $|L(H^*)|$, and (3) then uses the upper bound of $|L(H^*)|$ and the MBP with the most edges found so far to further tighten the lower bound of $|R(H^*)|$. Note that it would generate $O(\log (\delta_{max}^R+k))$ sets of constraints.

It would then run the \texttt{BB} algorithm for each set of constrains captured by $LB_L^i, UB_L^i, LB_R^i, UB_R^i$ in the order of $i = 1, 2, ...$. At the $i^{th}$ time, it utilizes the lower and upper bounds as follows. \underline{First}, it reduces the graph by computing the $(|LB_R^i - k|, |LB_L^i - k|)$-core of $G$ since according to~\cite{yu2021efficient,yu2022kbiplex}, any MBP with at least $|LB_L^i|$ vertices at the left side and $|LB_R^i|$ vertices at the right side must reside in the $(|LB_R^i - k|, |LB_L^i - k|)$-core of $G$. \underline{Second}, when it runs a \texttt{BB} algorithm, it prunes a branch $B = (S, C, D)$ if $|S_L| > |UB_L^i|$ or $|S_R| > UB_R^i$.

In summary, \texttt{PB} runs a \texttt{BB} algorithm multiple times, each time on a reduced graph. Hence, \texttt{PB} would boost the practical performance of a \texttt{BB} algorithm.

\subsection{Inclusion-Exclusion based Framework: \texttt{IE}}
\label{subsec:ie}

\texttt{IE} is adapted from the \emph{decomposition} technique, which has been widely used for enumerating and finding subgraph structures~\cite{wang2022listing,chen2021efficient,DBLP:conf/kdd/ConteMSGMV18,DBLP:conf/aaai/ZhouXGXJ20}.
The major idea of \texttt{IE} is to partition the graph into multiple ones (which may overlap) and run a \texttt{BB} algorithm on each of the subgraphs. Finally, it returns among the found MBPs the one with the most edges. Specifically, it partitions the graph $G$ to $|L|$ subgraphs, namely $G_i = (L_i, R_i, E_i)$ for $1\le i\le |L|$, as follows.
\begin{align}
L_i &= \Gamma_2(v_i,L) -  \{v_1,...,v_{i-1}\};\label{equation:subgraph-1}\\
R_i &= \bigcup_{v\in L_i}\Gamma(v,R);\label{equation:subgraph-2}\\
E_i &= \{ (v, u) | v\in L_i, u\in R_i, (v, u)\in E\}\label{equation:subgraph-3}
\end{align}
where $L = \{v_1, v_2, ..., v_{|L|}\}$ and $\Gamma_2(v_i,L)$ denotes the set of 2-hop neighbors of $v_i$ in $L$.
We note that the number of vertices in $G_i$, i.e., $|L_i| + |R_i|$, is bounded by $d^3$, where $d$ is maximum degree of the graph $G$.
We verify that the MaxBP $H^*$ must reside in one of the subgraphs formed as above (for which the proof could be found
\ifx \CR\undefined
{\LC in the appendix}
\else
{\LC in the technical report~\cite{TR}}
\fi
). Furthermore, the MBP found in $G_i$ would \emph{include} $v_i$ and \emph{exclude} $v_1, v_2, ..., v_{i-1}$.



Moreover, it prunes the following vertices from a subgraph $G_i$.
\begin{itemize}[leftmargin=*]
	\item $v\in L_i$ with $\delta(v, R_i)<\theta_R-k$ or $|\Gamma(v,R_i)\cap \Gamma(v_i,R)|<\theta_R-2k$;
	\item $u\in R_i$ with $\delta(u, L_i)<\theta_L-k$.
\end{itemize}
The correctness of pruning the vertices as above can be verified by contradiction with the size constraints based on $\theta_L$ and $\theta_R$ (the detailed proof can be found
\ifx \CR\undefined
{\LC in the appendix}
\else
{\LC in the technical report~\cite{TR}}
\fi
).

Finally, it runs the \texttt{BB} algorithm on each graph $G_i$ with some vertices pruned by starting from the branch $B_i = (S_i, C_i, D_i)$ with $S_i = \{v_i\}$, $D_i = \{v_1, v_2, ..., v_{i-1}\}$, and $C_i = V(G_i) - \{v_1, v_2, ..., v_i\}$. It then returns the MBP with the most edges among all MBPs found.

With the \texttt{IE} framework, the time complexity of a \texttt{BB} algorithm can be improved in certain cases. For example, it improves the time complexity of \texttt{FastBB} from {\revision $O(|V|\cdot d \cdot \gamma_k^{|V|})$} to $O(|L|\cdot d^3 \cdot \gamma_k^{d^3})$ for certain sparse graphs (e.g., those with $d^3 < |V|$).

\subsection{Combining \texttt{PB} and \texttt{IE}: \texttt{PBIE}}
\label{subsec:pbie}

We observe that the \texttt{PB} and \texttt{IE} frameworks can be naturally combined for our problem of finding the MaxBP. Specifically, we can first use \texttt{PB} to construct multiple reduced graphs $G_{i}$ with corresponding constraints of lower and upper bounds of the number of vertices at the left and right sides of a graph. Then, when for each reduced graph $G_i$, we further use \texttt{IE} to construct $|L(G_i)|$ subgraphs $G_i^{v_j}$ for $v_j \in L(G_i)$ with some vertices pruned if possible. Finally, we invoke a \texttt{BB} algorithm (e.g., \texttt{FastBB}) on each subgraph $G_i^{v_j}$ with the constraints and return the MBP with the most edges among all found MBPs. 
The pseudo-code of \texttt{PBIE} is presented in Algorithm~\ref{alg:Fa_BR}. 
%

\begin{algorithm}{}
\small
\caption{Combine \texttt{PB} and \texttt{IE}: \texttt{PBIE} (for \texttt{FastBB})}
\label{alg:Fa_BR}
\KwIn{A graph $G(L,R,E)$, $k$, $\theta_L$ and $\theta_R$}
\KwOut{The maximal $k$-biplex $H^*$ with the most edges}
$H^*_0\leftarrow \emptyset$; $LB_L^0\leftarrow \delta_{max}^R+k$; $i\leftarrow 1$\;
\While{true}{
    Set $LB_L^i$, $UB_L^i$, $LB_R^i$, $UB_R^i$ according to Equations~(\ref{equation:bounds-left}) and (\ref{equation:bounds-right})\;
    \If{$UB_L^i \leq \theta_L$}
    {
        \textbf{return} $H^*_{i-1}$;
    }
    %
    %
    Compute a reduced graph $G_{i}$ as the $(|LB_R^i - k|, |LB_L^i - k|)$-core of $G$\;
    %
    $H_{i}^*\leftarrow H_{i-1}^*$\;
    \For{$v_j\in L(G_{i})$}{
       %
        Construct a subgraph $G^{v_j}_{i}$ based on Equations~(\ref{equation:subgraph-1} - \ref{equation:subgraph-3}) with some vertices further pruned\;
        $H^*_{i}\!\leftarrow$ invoke a \texttt{FastBB} algorithm with $G_{i}^{v_j}$, $k$, the lower/upper bounds, and $H^*_{i}$\;
    }
    $i\leftarrow i+1$\;
}
\textbf{return} $H^*_{i-1}$\;
\end{algorithm}

{\roundA
\noindent\textbf{Time complexity.} 
The time cost is dominated by part of invoking \texttt{FastBB} (line 3-11).
There are at most $O(\log (\delta_{max}^R+k))$ iterations (line 3-11). For each iteration, it constructs at most $O(|L|)$ subgraphs.  
Hence, \texttt{FastBB} is invoked by at most $O(\log (\delta_{max}^R+k) \cdot |L|)$ times. 
Besides, the number of vertices in $G_{i+1}^{v_j}$ is bounded by $O(d^3)$.
Therefore, the time complexity of \texttt{PBIE} (when used for boosting \texttt{FastBB}) is $O(\log (\delta_{max}^R+k)\cdot |L| \cdot d^4\cdot \gamma_k^{d^3})$ where $\gamma_k$ is a real number strictly smaller than $2$ (refer to Theorem \ref{lemma:time_complexity} for details).
We remark that the large graphs in real applications are usually sparse and have $d$ far smaller than the total number of vertices.
}


\section{Experiments}
\label{sec:exp}
\begin{table*}[h]
	\centering
	\vspace{-0.10in}
	\scriptsize
	\small
	\caption {Real datasets}
	\label{tab:dataset}
 	\vspace{-0.15in}
	\begin{tabular}{c|c|r|r|r|c|c|c|c}
		\hline
		 \multicolumn{1}{c|}{{Dataset}}
		 & \multicolumn{1}{c|}{{Category}}
		 & \multicolumn{1}{c|}{$|L|$}
		 & \multicolumn{1}{c|}{$|R|$}
		 & \multicolumn{1}{c|}{$|E|$}
		 & \multicolumn{1}{c|}{\revision Density}
		 & \multicolumn{1}{c|}{\revision $d$}
		 & \multicolumn{1}{c|}{\revision $|V(H)|$}
		 & \multicolumn{1}{c}{\revision $|E(H)|$}
		\\
		\hline\hline
		Divorce & Feature  & 9  & 50  & 225 & \revision 7.62 & \revision 37 & \revision 23 & \revision 87  \\
		\hline
		Cities & Feature  & 46  & 55 & 1342 &\revision 26.4 &\revision 54 &\revision 43 &\revision 437 \\
		\hline
		Cfat & 	Biology  & 200  & 200  & 1537 &\revision 7.66 &\revision 16 &\revision 13 &\revision 41 \\
		\hline
		Opsahl & Authorship  & 2,865  & 4,558  & 16,910 &\revision 4.58 &\revision 116 &\revision 19 &\revision 87 \\
		\hline
		\hline
		\textbf{Writer} & Authorship  & 89,356  & 46,213 & 144,340 &\revision 2.12 &\revision 246 &\revision 38 &\revision 136  \\
		\hline
		YouTube & Affiliation  & 94,238  & 30,087  & 293,360 &\revision 4.72 &\revision 7,591 &\revision 319 &\revision 945  \\
		\hline
		\textbf{Location} & Feature  & 172,091  & 53,407  & 293,697 &\revision 2.60 &\revision 12,189 &\revision 287 &\revision 849 \\
		\hline
		Actors & Affiliation   & 127,823  & 383,640  & 1,470,404 &\revision 5.64 &\revision 646 &\revision 117 &\revision 341  \\
		\hline
		IMDB & Affiliation  & 303,617  & 896,302  & 3,782,463 &\revision 5.72 &\revision 1,590 &\revision 192 &\revision 564  \\
		\hline
		\textbf{DBLP} & Authorship  & 1,425,813  & 4,000,150 & 8,649,016 &\revision 3.18 &\revision 951 &\revision 127 &\revision 369 \\
		\hline
		Amazon & Rating   & 6,703,391  & 957,764  & 12,980,837 &\revision 3.28 &\revision 9,990 &\revision 2,255 &\revision 1,1248 \\
		\hline
	    \textbf{Google} & Hyperlink  & 17,091,929  & 3,108,141 & 14,693,125 &\revision 1.46 &\revision 2,318 &\revision 44 &\revision 120   \\
		\hline
	\end{tabular}
\end{table*}

\noindent\textbf{Datasets.} We use both real and synthetic datasets in our experiments. 
{\revision The real datasets and their statistics are summarized in Table \ref{tab:dataset} (http://konect.cc/) where the edge density of a bipartite graph $G(L\cup R, E)$ is defined as $2\times |E|/({|L|+ |R|})$ and $d$ denotes the maximum degree. 
Note that these real graphs are usually sparse overall yet with some dense subgraphs (e.g., the maximum degree at least 16 and up to 12,189) since they follow the power-low distribution.} 
The Erd{\"o}s-R{\'e}yni (ER) synthetic datasets are generated by first creating a certain number of vertices and then randomly adding a certain number of edges between pairs of vertices.
We set the number of vertices and edge density as 100k and 20 for synthetic datasets, respectively, by default.

\smallskip\noindent\textbf{Baselines.} We compare our algorithm \texttt{PBIE+FastBB} with three baselines, namely \texttt{iMBadp}~\cite{yu2021efficient}, \texttt{FPadp}~\cite{DBLP:conf/aaai/ZhouXGXJ20} and \texttt{iTradp}~\cite{yu2022kbiplex}. 
Specifically, \texttt{PBIE+FastBB} adopts the combined framework \texttt{PBIE} and employs the improved algorithm \texttt{FastBB} within the framework. \texttt{iMBadp} and \texttt{iTradp} correspond to the adaptions of existing algorithms designed for enumerating \emph{large MBPs}, namely \texttt{iMB}~\cite{yu2021efficient} and \texttt{iTraversal}~\cite{yu2022kbiplex}.
{\revision Specifically, \texttt{iMBadp} adopts the BK branching for each side of a bipartite graph and is equipped with the pruning techniques developed in this paper. 
\texttt{iTradp} follows a \emph{reverse search}~\cite{DBLP:journals/jcss/CohenKS08} method and cannot be equipped with the pruning techniques 
developed in Section~\ref{subsec:fastbb} (since they are not compatible with the reverse search framework).
\texttt{FPadp} corresponds to a branch-and-bound algorithm with the hybrid branching strategy~\cite{DBLP:conf/aaai/ZhouXGXJ20} and the pruning techniques developed in this paper. 
Furthermore, \texttt{FPadp} inherits \texttt{FaPlexen}~\cite{DBLP:conf/aaai/ZhouXGXJ20} and employs the \texttt{IE} framework.}

{\revision
%
We remark that all baseline methods are partial enumeration based methods, i.e., they do not enumerate all MBPs for finding the MaxBPs. This is because
(1) their original versions were designed to enumerate only those MBPs with sizes at least some threshold but not all MBPs (as specified in the introduction section) and (2) they are enhanced by the four pruning conditions introduced in this paper as much as possible for pruning MBPs that are not MaxBPs (as specified above).
}


\smallskip
\noindent\textbf{Settings.} All algorithms were written in C++ and run on a machine with a 2.66GHz CPU and 32GB main memory running CentOS. We set the time limit (INF) as 24 hours and use 4 representative datasets as default ones, i.e., Writer, Location, DBLP and Google, which cover various graph scales. We set both $\theta_L$ and $\theta_R$ as $2k+1$ and both $K$ and $k$ as 1 by default.
Our code, data and additional experimental results are available at https://github.com/KaiqiangYu/SIGMOD23-MaxBP.


\subsection{Comparison among algorithms}

\noindent\textbf{All datasets.} We compare all algorithms on various datasets and show the running time in Figure~\ref{fig:all_datasets}. {\revision Besides, we report the number of vertices and the number of edges, denoted by $V(H)$ and $|E(H)|$, respectively, within the returned MaxBP in Table~\ref{tab:dataset}. We have the following observations. First, the found MaxBPs are usually sufficiently large to be meaningful (e.g., with at least 41 and up to 11,248 edges). }
Second, \texttt{PBIE+FastBB} outperforms all other algorithms on all datasets. This is consistent with our theoretical analysis that the worst-case running time of \texttt{PBIE+FastBB} is smaller than that of other algorithms. 
Third, \texttt{PBIE+FastBB} can handle all datasets within INF while others cannot finish on large datasets, e.g., Amazon and Google, which demonstrates its scalability. This is mainly because the framework \texttt{PBIE} would quickly locate the MaxBP at several smaller subgraphs while dramatically pruning many unfruitful vertices so as to reduce the search space.

\begin{figure}[]
	\centering
 	\vspace{-0.10in}
	\includegraphics[width=0.98\linewidth]{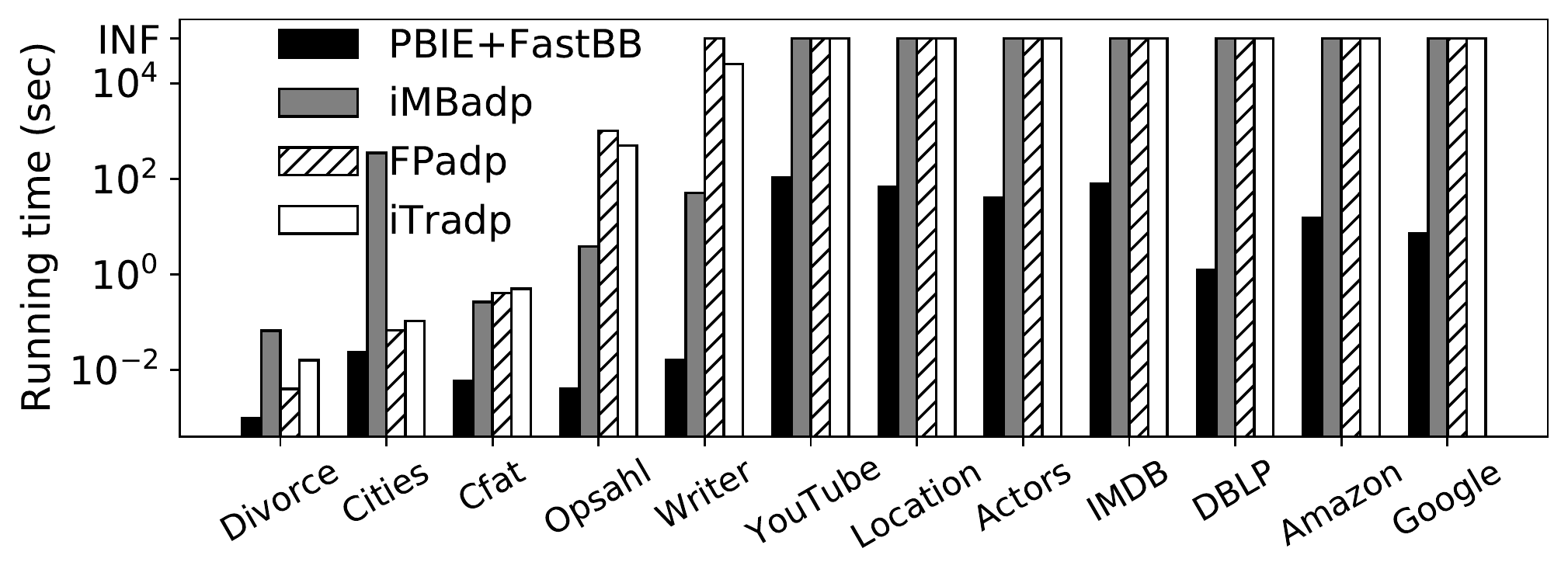}
	\vspace{-0.15in}
	\caption{Comparison on all real datasets ($k=1$)}
	\label{fig:all_datasets}
\end{figure}
\begin{figure}[]
	\centering
	\begin{tabular}{c c}
		
		\begin{minipage}{3.80cm}
			\includegraphics[width=4.1cm]{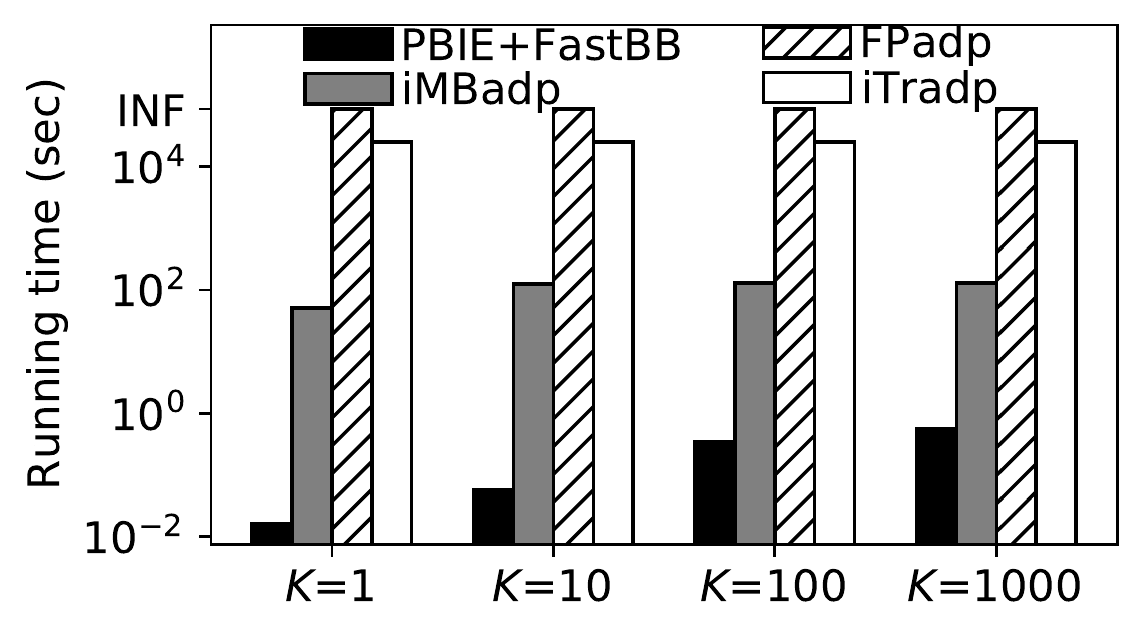}
		\end{minipage}
		&
		\begin{minipage}{3.80cm}
			\includegraphics[width=4.1cm]{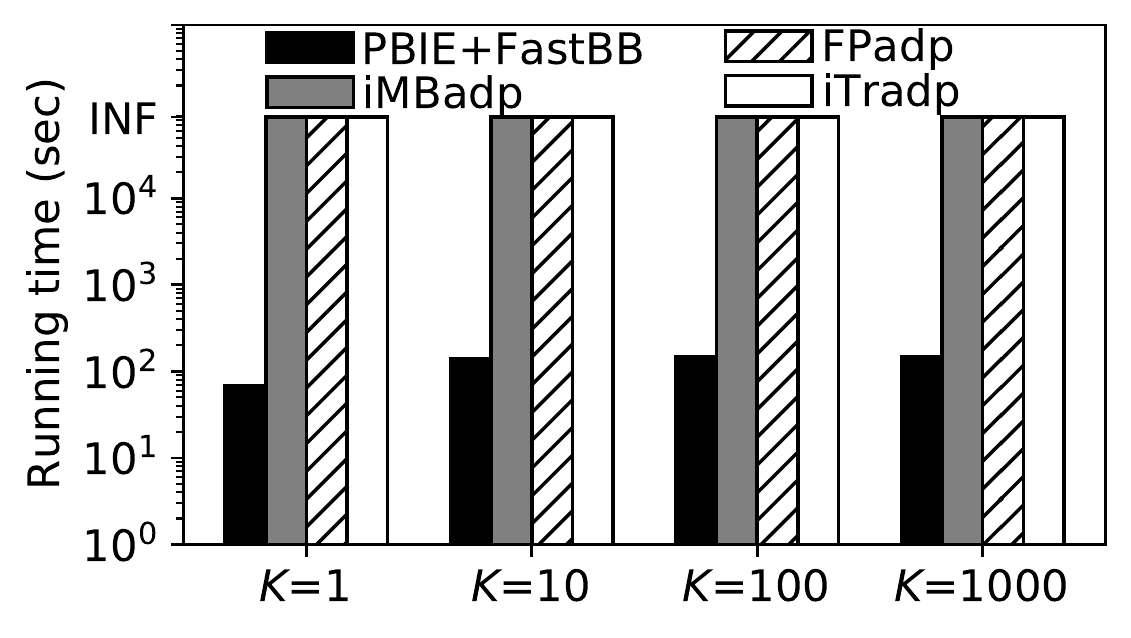}
		\end{minipage}		
		\\
		(a) Varying $K$ (Writer)
		&
		(b) Varying $K$ (Location)
		\\
		\begin{minipage}{3.80cm}
			\includegraphics[width=4.1cm]{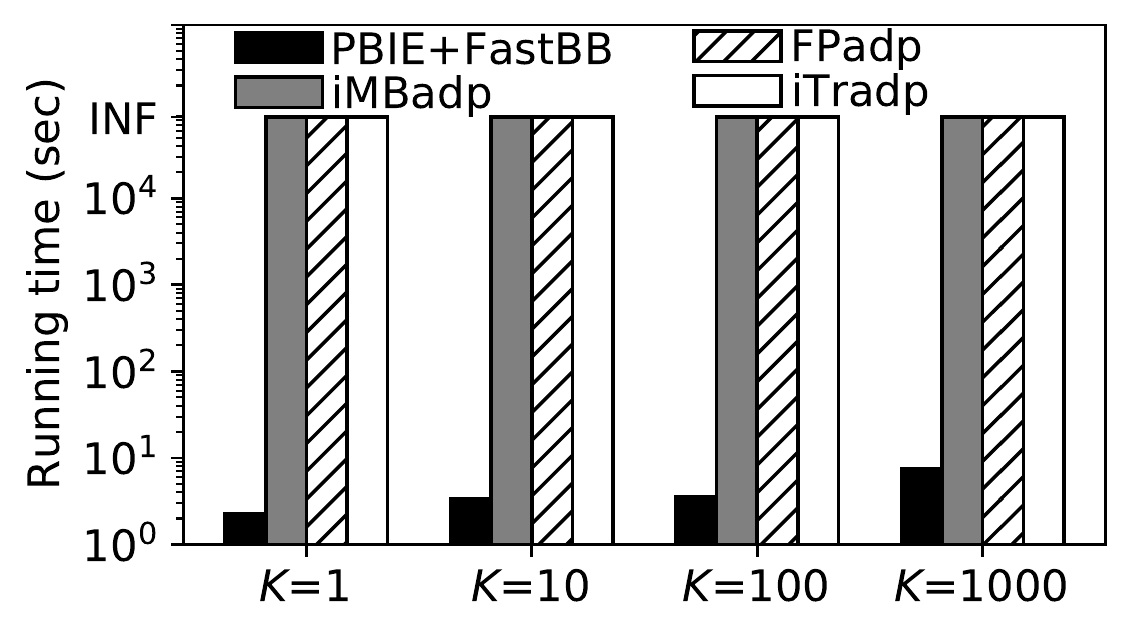}
		\end{minipage}
		&
		\begin{minipage}{3.80cm}
			\includegraphics[width=4.1cm]{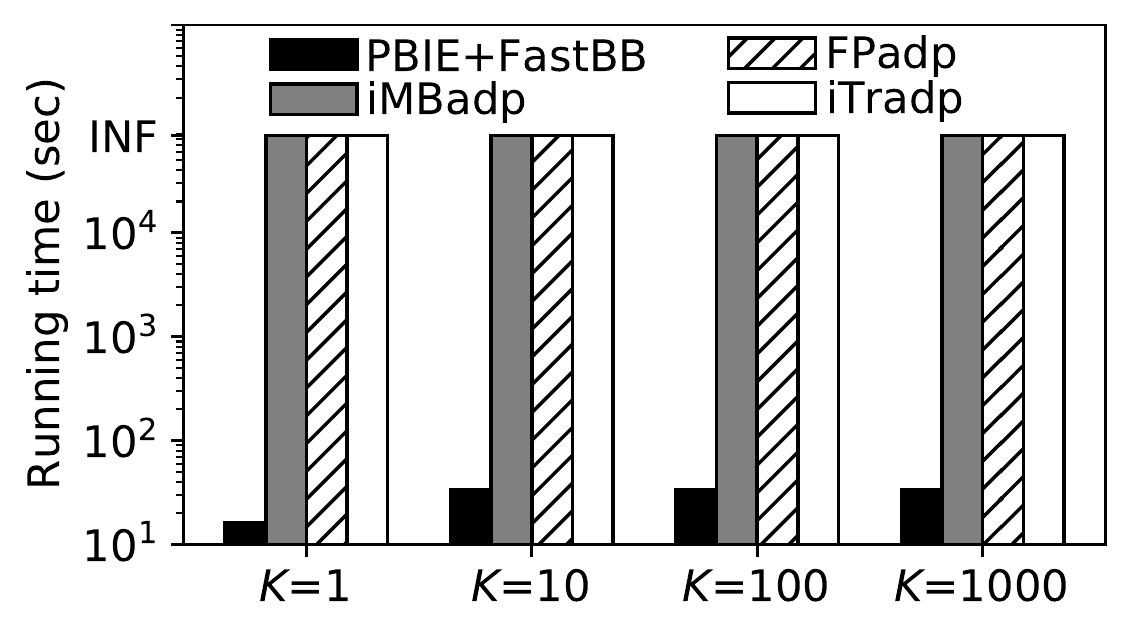}
		\end{minipage}		
		\\
		(c) Varying $K$ (DBLP)
		&
		(d) Varying $K$ (Google) 
	\end{tabular}
	\vspace{-0.15in}
	\caption{Comparison by varying $K$ ($k=1$)}
	\label{fig:vary_K}
\end{figure}

\smallskip
\noindent\textbf{Varying $K$.} The results of finding {\roundA $K$ MaxBPs} are shown in Figure~\ref{fig:vary_K}. \texttt{PBIE+FastBB} outperforms other algorithms by around 2-5 orders of magnitude. {\roundA Besides, it has the running time clearly rise as $K$ increases compared to other algorithms. Possible reasons include (1) \texttt{FastBB} would explore more MBPs to find the $K$ MaxBPs as $K$ grows and (2) \texttt{PBIE} becomes less effective as $K$ grows. In addition, \texttt{iTradp} has the running time almost not changed with $K$ since it needs to explore almost all MBPs without any powerful pruning.} 

\smallskip
\noindent\textbf{Varying $k$.} The results are shown in Figure~\ref{fig:vary_k}. \texttt{PBIE+FastBB} significantly outperforms other algorithms by up to five orders of magnitude. In addition, it has the running time first increase and then decrease as $k$ grows. Possible reasons include: (1) the number of $k$-biplexes increases exponentially with $k$, which causes the increase of the running time; (2) the thresholds $\theta_L$ and $\theta_R$ (i.e., $2k+1$ by default) increase with $k$ and correspondingly the pruning techniques that are based on $\theta_L$'s and $\theta_R$ become more effective. 

\begin{figure}[]
	\centering
	\begin{tabular}{c c}
		
		\begin{minipage}{3.80cm}
			\includegraphics[width=4.1cm]{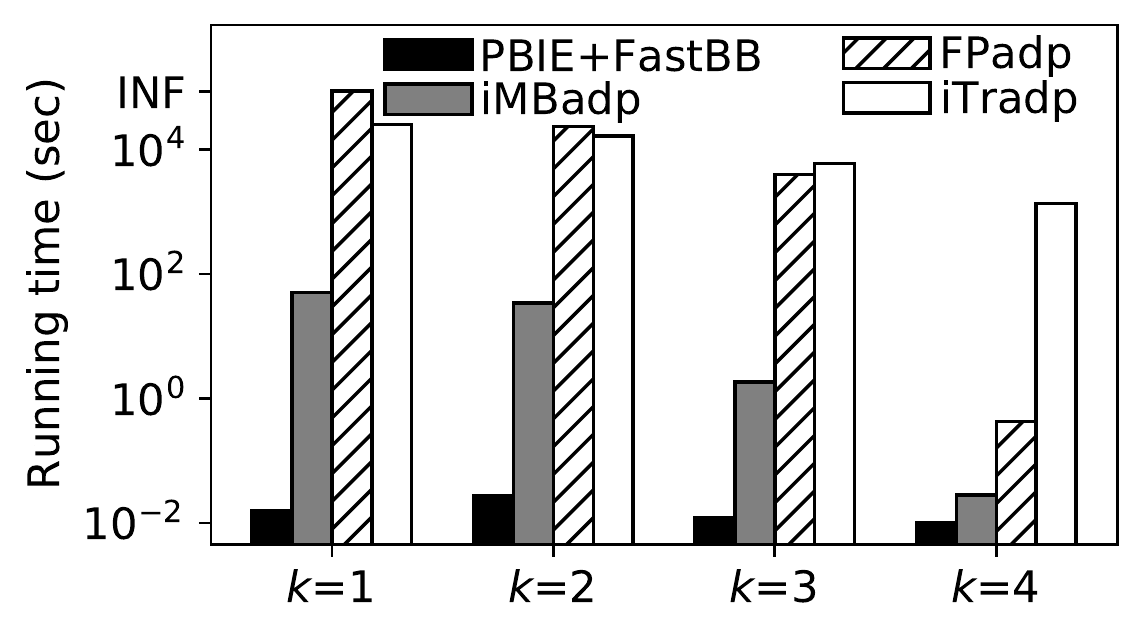}
		\end{minipage}
		&
		\begin{minipage}{3.80cm}
			\includegraphics[width=4.1cm]{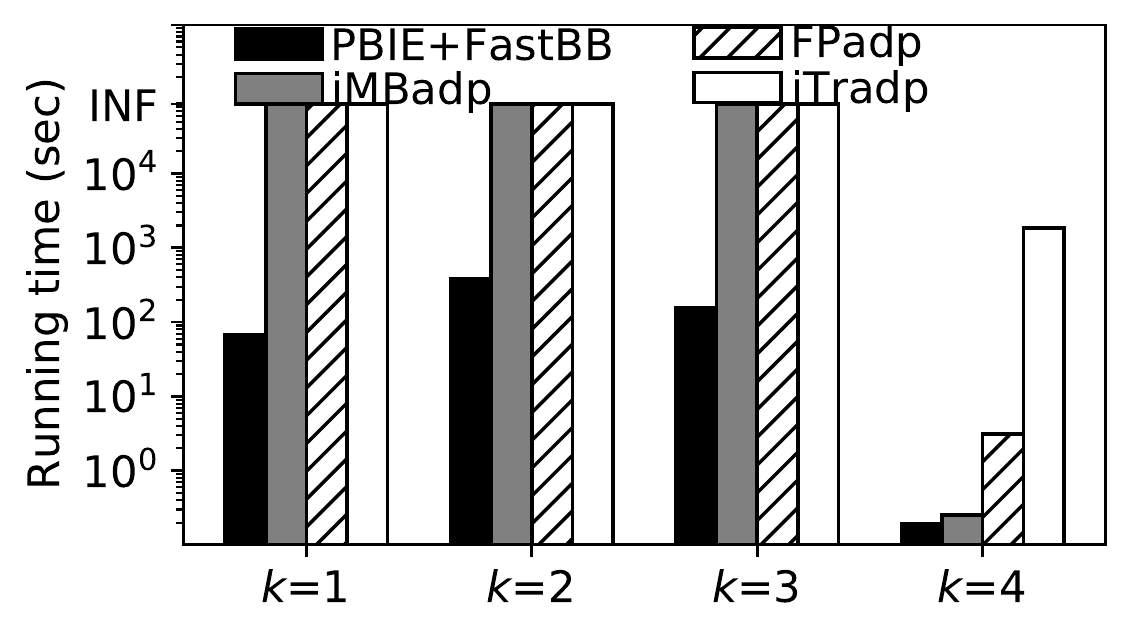}
		\end{minipage}		
		\\
		(a) Varying $k$ (Writer)
		&
		(b) Varying $k$ (Location)
		\\
		\begin{minipage}{3.80cm}
			\includegraphics[width=4.1cm]{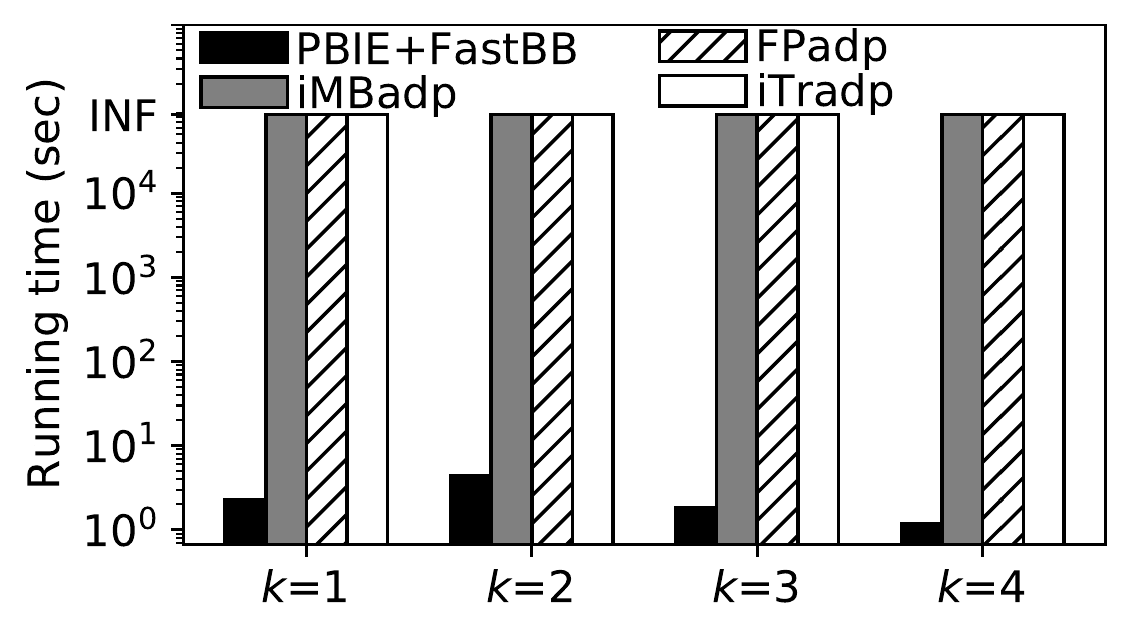}
		\end{minipage}
		&
		\begin{minipage}{3.80cm}
			\includegraphics[width=4.1cm]{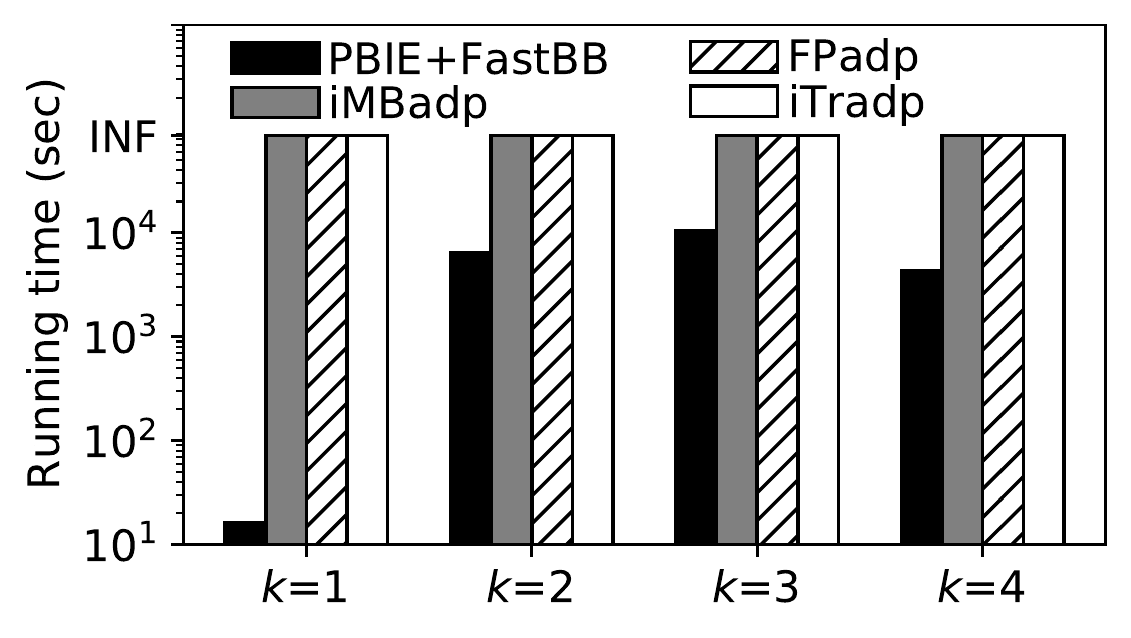}
		\end{minipage}		
		\\
		(c) Varying $k$ (DBLP)
		&
		(d) Varying $k$ (Google) 
	\end{tabular}
	\vspace{-0.15in}
	\caption{Comparison by varying $k$}
	\label{fig:vary_k}
\end{figure}

\smallskip
\noindent\textbf{Varying $\theta_L$ and $\theta_R$ thresholds.} 
The results are shown in Figure~\ref{fig:vary_theta}. \texttt{PBIE+FastBB} outperforms all other algorithms by achieving up to around 1000$\times$ speedup. Besides, the running time of all algorithms decreases as $\theta_L$ and $\theta_R$ grow. The reason is two-fold: (1) the search space (e.g., the number of large $k$-biplexes with the size of each side at least $\theta_L$ and $\theta_R$) decreases as $\theta_L$ and $\theta_R$ grow and (2) the pruning rules are more effective for larger $\theta_L$'s and $\theta_R$'s. 
\begin{figure}[]
	\centering
	\begin{tabular}{c c}
		
		\begin{minipage}{3.80cm}
			\includegraphics[width=4.1cm]{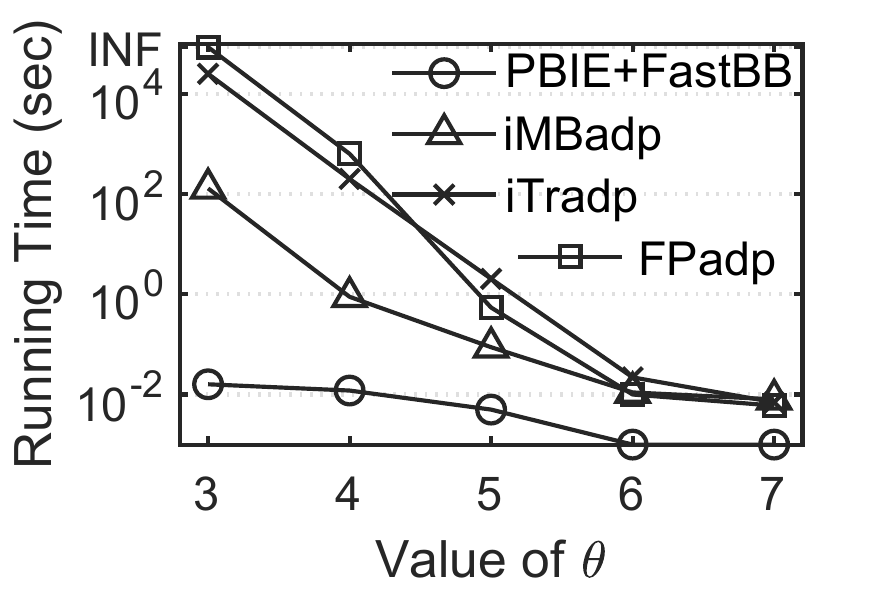}
		\end{minipage}
		&
		\begin{minipage}{3.80cm}
			\includegraphics[width=4.1cm]{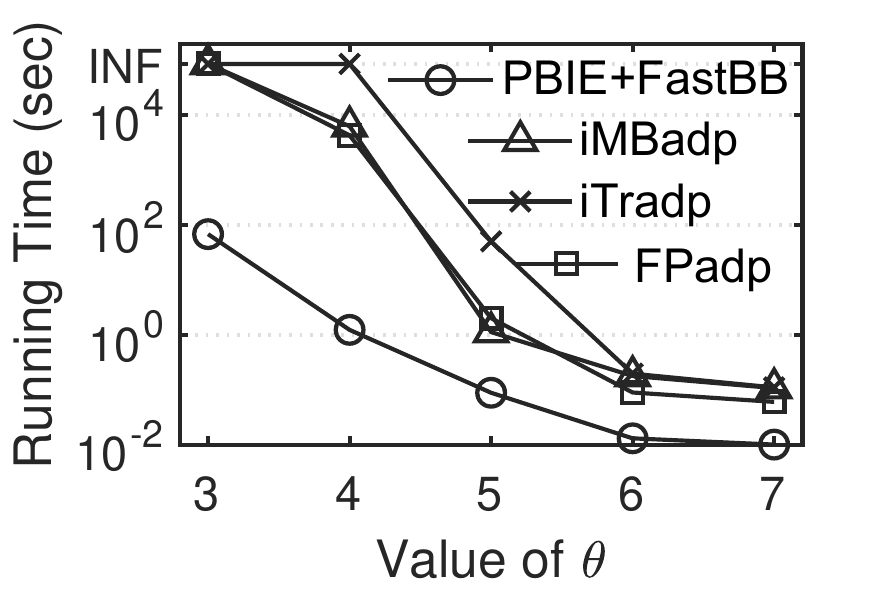}
		\end{minipage}		
		\\
		(a) Varying $\theta$ (Writer)
		&
		(b) Varying $\theta$ (Location)
		\\
		\begin{minipage}{3.80cm}
			\includegraphics[width=4.1cm]{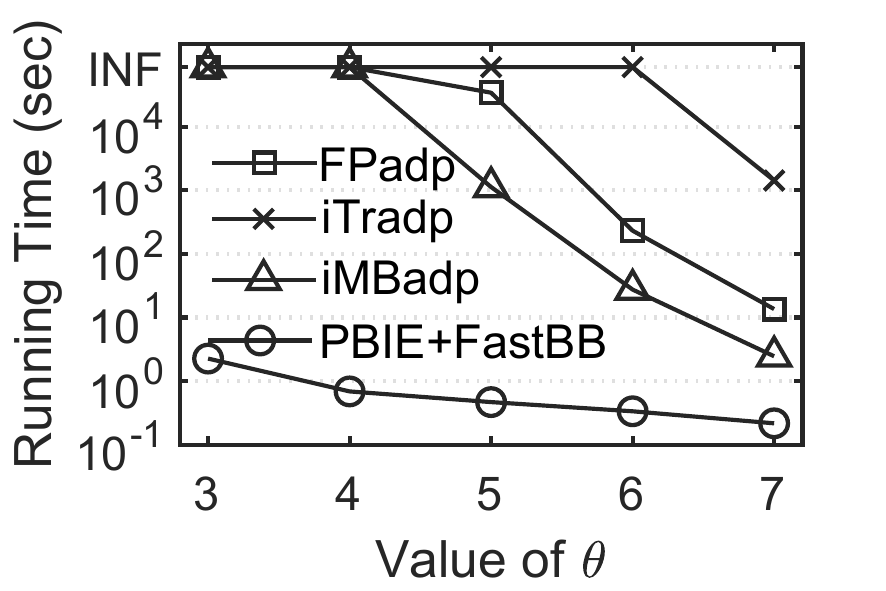}
		\end{minipage}
		&
		\begin{minipage}{3.80cm}
			\includegraphics[width=4.1cm]{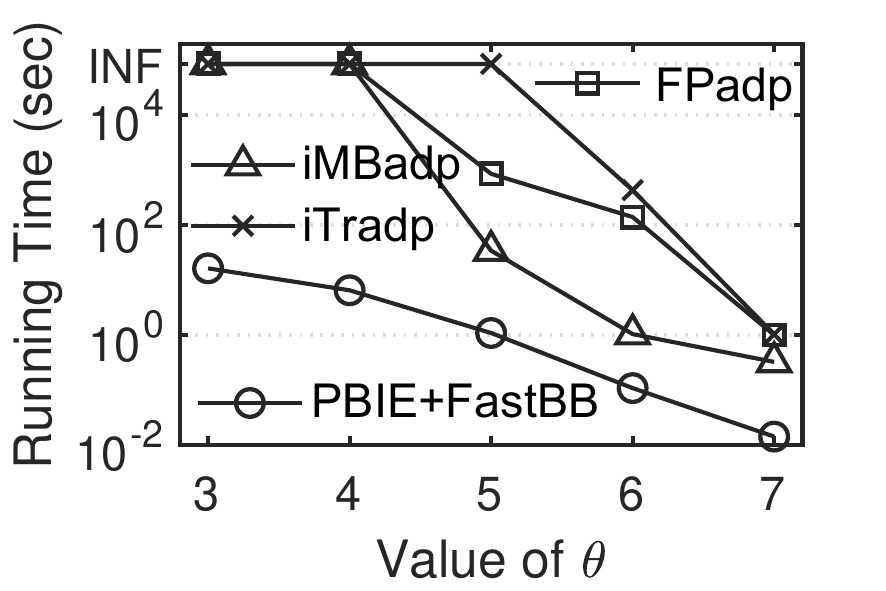}
		\end{minipage}		
		\\
		(c) Varying $\theta$ (DBLP)
		&
		(d) Varying $\theta$ (Google) 
	\end{tabular}
	\vspace{-0.15in}
	\caption{Comparison by varying $\theta=\theta_L=\theta_R$ ($k=1$)}
	\label{fig:vary_theta}
\end{figure}

\smallskip
\noindent\textbf{Varying \# of vertices (synthetic datasets).} The results are shown in Figure~\ref{fig:synthetic}(a). \texttt{PBIE+FastBB} outperforms other algorithms by achieving at least 10$\times$ speedup and can handle the largest datasets with 1 million vertices and 10 million edges within INF while others cannot. Besides, the speedup increases as the graph become larger. 
{\roundA This is mainly because $\texttt{PBIE}$ would prune more unfruitful vertices when locating the MaxBP in several smaller subgraphs. In addition, the results are well aligned with the theoretical results, i.e., the worst-case time complexity of \texttt{PBIE+FastBB} is exponential wrt $d^3$ while that of others is exponential wrt $|V|$. Hence, \texttt{PBIE+FastBB} has larger speed-ups as the graph scale becomes larger (where $d$ remains almost the same due to the fixed edge density, i.e., 20).}

\smallskip
\noindent\textbf{Varying edge density (synthetic datasets).} The results are shown in Figure~\ref{fig:synthetic}(b). \texttt{PBIE+FastBB} achieves at least 10$\times$ speedup compared with other algorithms. In addition, the speedup decreases as the graph becomes denser. {\roundA The reason is that the maximum degree of the bipartite graph, i.e., $d$, increases as the graph becomes denser. 
}

\begin{figure}[]
	\centering
	\vspace{-0.10in}
	\begin{tabular}{c c}
		
		\begin{minipage}{3.80cm}
			\includegraphics[width=4.1cm]{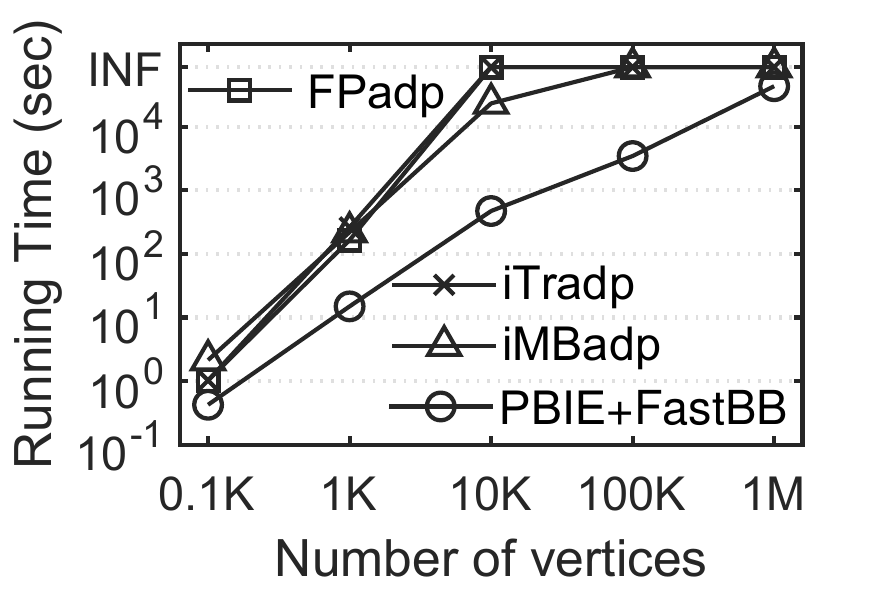}
		\end{minipage}
		&
		\begin{minipage}{3.80cm}
			\includegraphics[width=4.1cm]{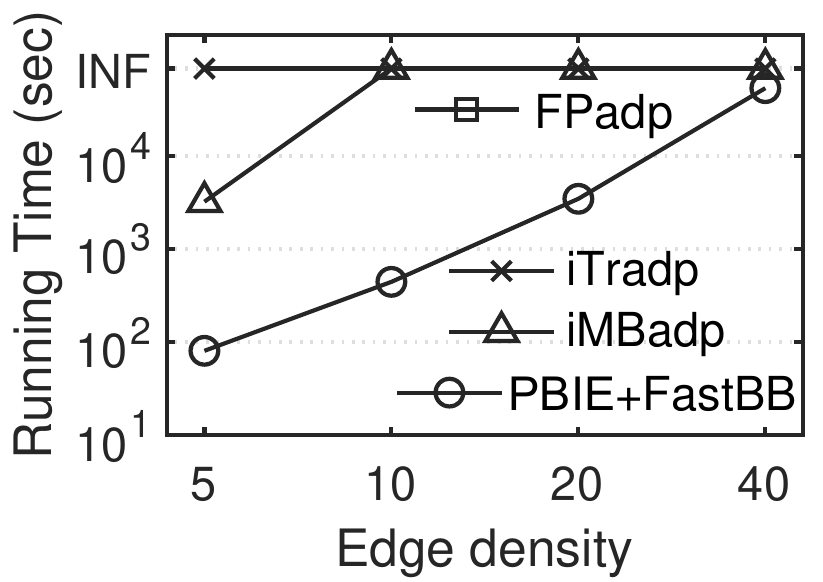}
		\end{minipage}		
		\\
		(a) Varying \# of vertices
		&
		(b) Varying edge density
	\end{tabular}
	\vspace{-0.15in}
	\caption{Comparison on synthetic datasets ($k=1$)}
	\label{fig:synthetic}
\end{figure}

\subsection{Performance study}

\noindent\textbf{Comparison among various enumeration schemes.} We study the effect of various enumeration schemes, namely \texttt{FastBB}, \texttt{BasicBB}, \texttt{FPadp}, \texttt{iMBadp}, and \texttt{iTradp}. We note that all of them are run with the framework \texttt{PBIE} for fair comparison. 
{\roundA In particular, \texttt{BasicBB} adopts the classic BK branching and uses the non-decreasing vertex ordering}.
The results are shown in Figure~\ref{fig:exp_alg}(a) and (b) for varying $k$, and (c) and (d) for varying $\theta=\theta_L=\theta_R$. 
First, \texttt{FastBB} outperforms other algorithms, which demonstrates the superiority of proposed Sym-BK branching scheme. {\roundA Besides, the achieved speedup decreases with $\theta$ since the search space (e.g., the number of large MBPs with the size of each side at least $\theta_L$ and $\theta_R$) decreases as $\theta_L$ and $\theta_R$ grow.}
Second, the algorithms following the BK branching perform better than \texttt{iTradp} that follows a reverse search method. The reason is that the later one cannot be enhanced with effective pruning rules as branch-and-bound algorithms.
%

\begin{figure}[]
	\centering
	\begin{tabular}{c c}
		
		\begin{minipage}{3.80cm}
			\includegraphics[width=4.1cm]{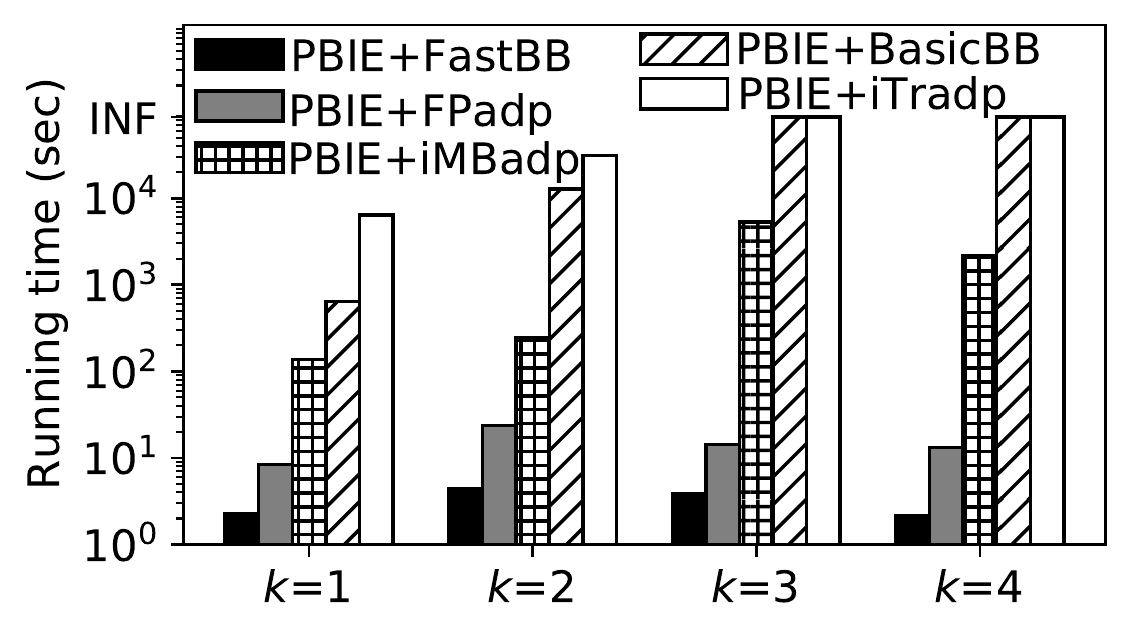}
		\end{minipage}
		&
		\begin{minipage}{3.80cm}
			\includegraphics[width=4.1cm]{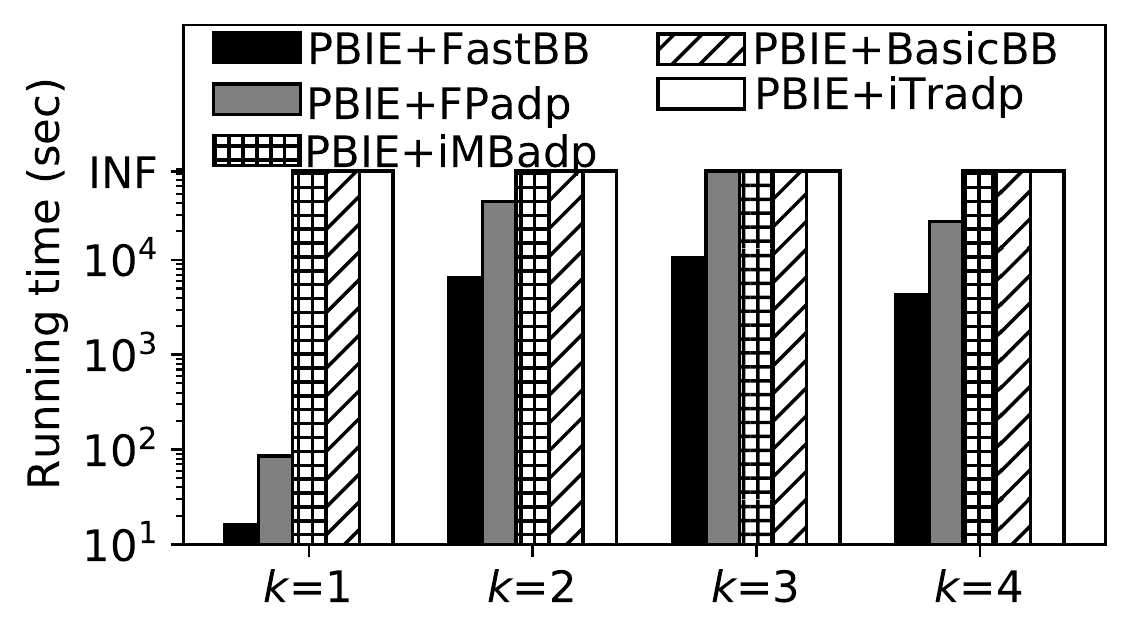}
		\end{minipage}		
		\\
		(a) Varying $k$ (DBLP)
		&
		(b) Varying $k$ (Google)
		\\
		\begin{minipage}{3.80cm}
			\includegraphics[width=4.1cm]{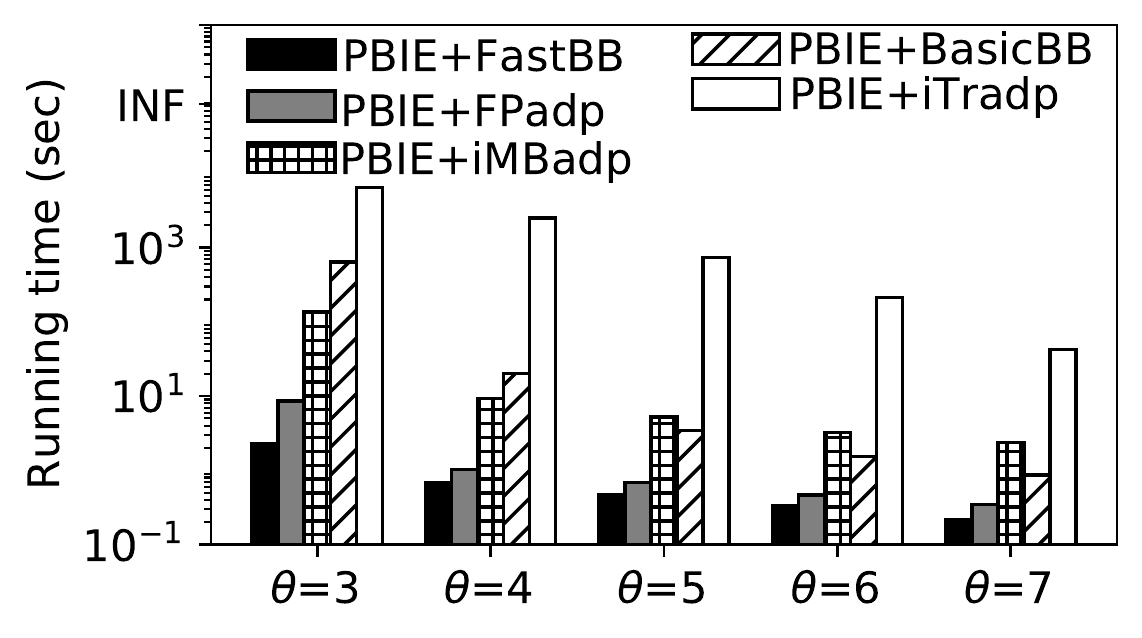}
		\end{minipage}
		&
		\begin{minipage}{3.80cm}
			\includegraphics[width=4.1cm]{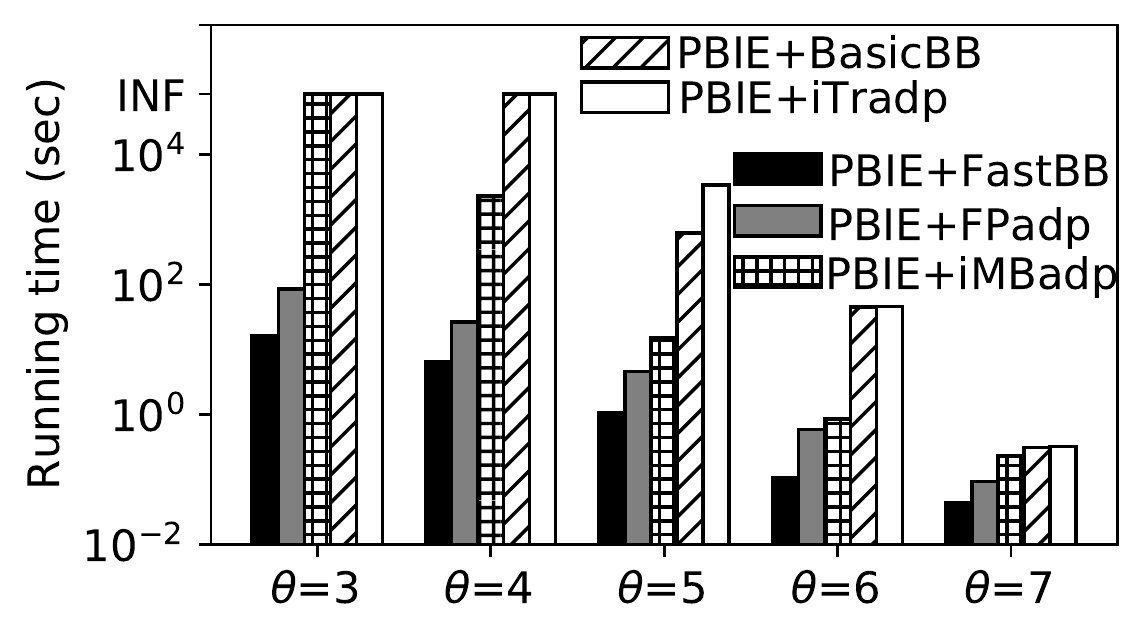}
		\end{minipage}		
		\\
		\ \ \ (c) Varying $\theta$ (DBLP, $k$=1)
		&
		\ \ \ (d) Varying $\theta$ (Google, $k$=1)
	\end{tabular}
	\vspace{-0.15in}
	\caption{Comparison among various enumeration schemes}
	\label{fig:exp_alg}
	\vspace{-0.15in}
\end{figure}

\smallskip
\noindent\textbf{Comparison among frameworks.} We study the effect of different frameworks and compare four different versions, namely (1) \texttt{FastBB}: without any framework, (2) \texttt{IE+FastBB}: with adapted framework \texttt{IE}, (3) \texttt{PB+FastBB}: with adapted progressive bounding framework \texttt{PB}, and (4) \texttt{PBIE+FastBB}: with proposed combined framework. We note that all frameworks adopt \texttt{FastBB} for fair comparison. The results are shown in Figure~\ref{fig:exp_framework}(a) and (b) for varying $k$, and (c) and (d) for varying $\theta=\theta_L=\theta_R$. 
First, both \texttt{IE+FastBB} and \texttt{PB+FastBB} outperform \texttt{FastBB}, which demonstrates the efficiency and scalability of adapted frameworks. Moreover, \texttt{PBIE+FastBB} performs the best and can handle all datasets and settings within INF. {This is because \texttt{PBIE} can significantly reduce the size of the original graph (details can be found
\ifx \CR\undefined
{\LC in the appendix}
\else
{\LC in the technical report~\cite{TR}}
\fi
).
}
Second, \texttt{PB+FastBB} performs better than \texttt{IE+FastBB}. This is because \texttt{PB} can quickly locate the MaxBP at a much smaller subgraph.
Third, the speedup decreases as $\theta$ grows. As discussed earlier, the search space gets smaller with $\theta$ and thus the frameworks would have less effects on the running time.

\begin{figure}[]
	\centering
	\vspace{-0.15in}
	\begin{tabular}{c c}
		
		\begin{minipage}{3.80cm}
			\includegraphics[width=4.1cm]{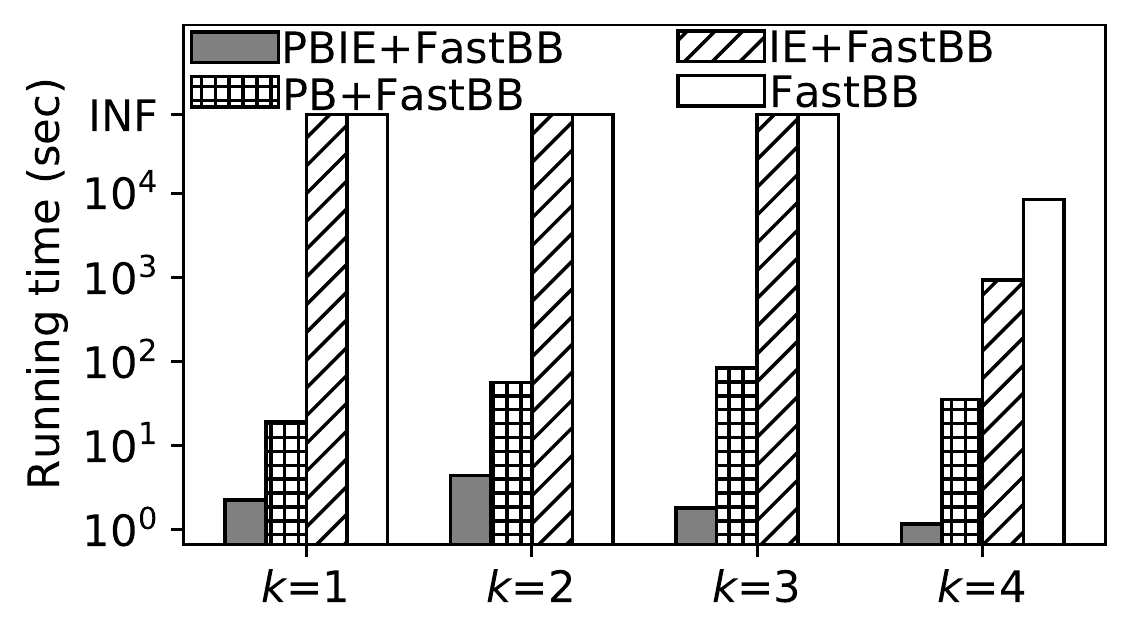}
		\end{minipage}
		&
		\begin{minipage}{3.80cm}
			\includegraphics[width=4.1cm]{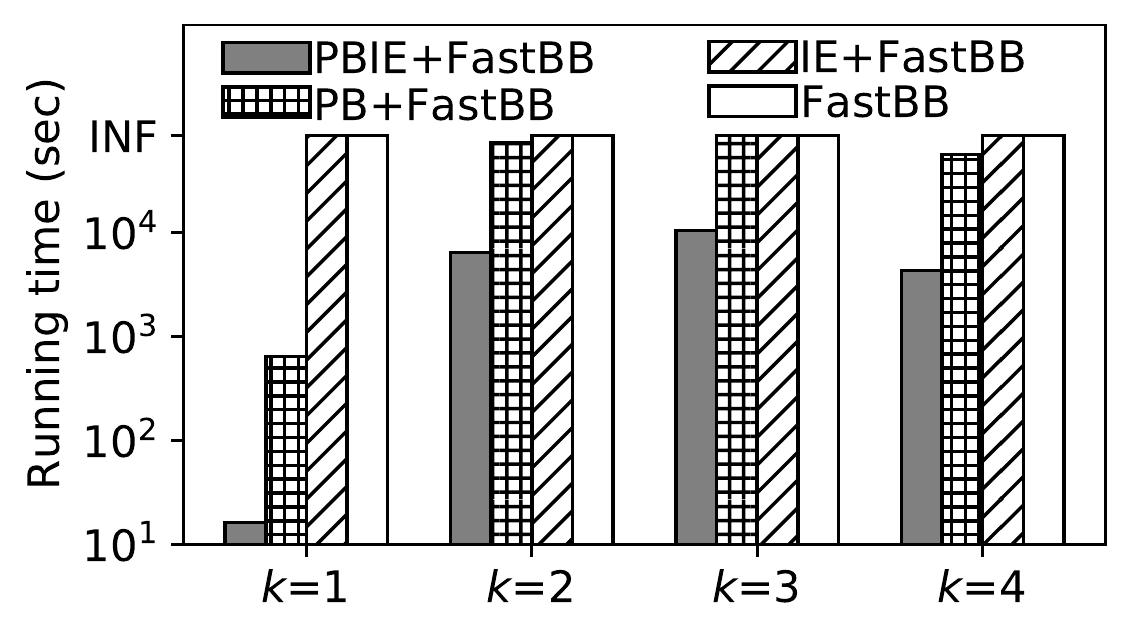}
		\end{minipage}		
		\\
		(a) Varying $k$ (DBLP)
		&
		(b) Varying $k$ (Google)
		\\
		\begin{minipage}{3.80cm}
			\includegraphics[width=4.1cm]{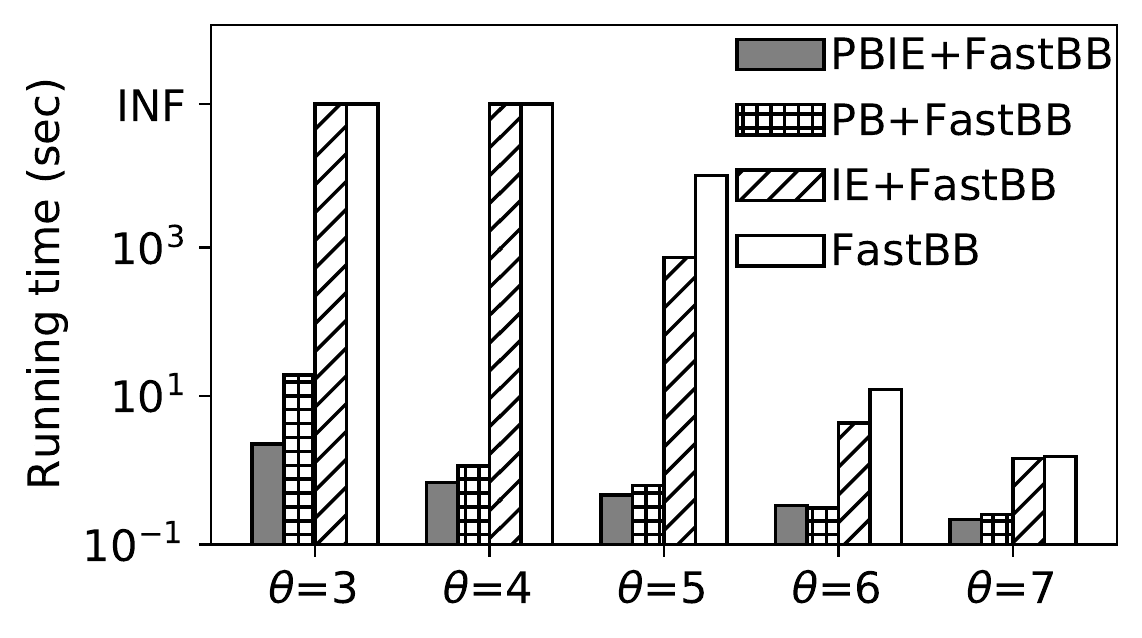}
		\end{minipage}
		&
		\begin{minipage}{3.80cm}
			\includegraphics[width=4.1cm]{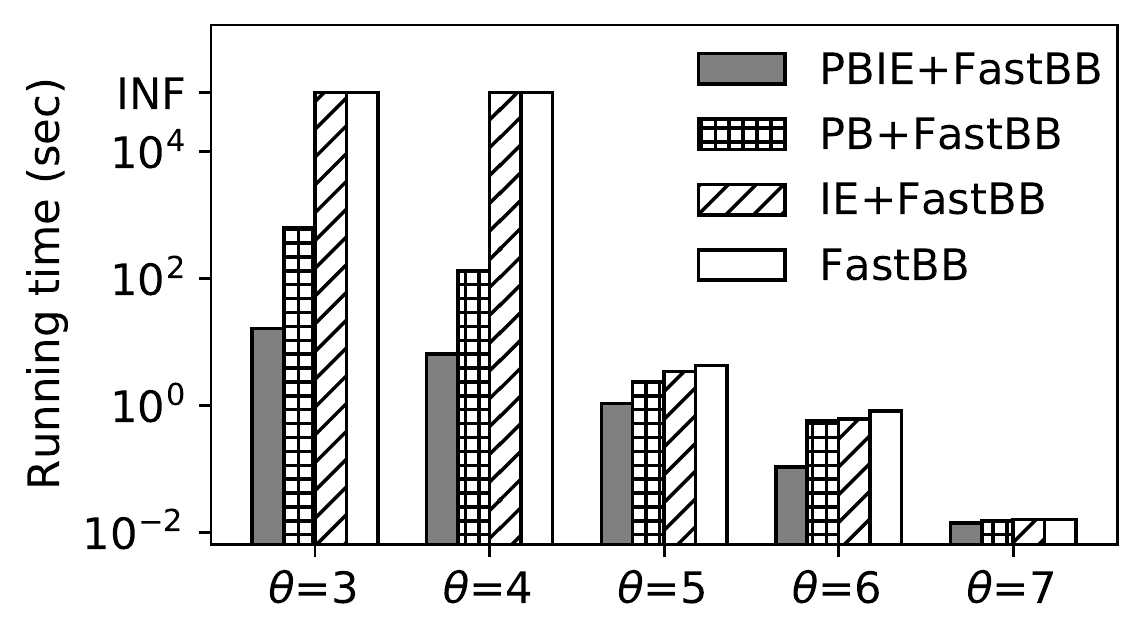}
		\end{minipage}		
		\\
		\ \ \ (c) Varying $\theta$ (DBLP, $k$=1)
		&
		\ \ \ (d) Varying $\theta$ (Google, $k$=1)
	\end{tabular}
	\vspace{-0.15in}
	\caption{Comparison among frameworks}
	\label{fig:exp_framework}
	\vspace{-0.10in}
\end{figure}
\subsection{Case study: Fraud Detection}
We investigate two cohesive models, namely maximal $k$-biplex and maximum $k$-biplex, for a fraud detection task~\cite{DBLP:conf/kdd/HooiSBSSF16} on the Software dataset ~\cite{dataset}.
The dataset contains 459,436 reviews on 21,663 softwares by 375,147 users. We consider the random camouflage attack scenario~\cite{DBLP:conf/kdd/HooiSBSSF16} where a fraud block with 1K fake users (``K'' means a thousand), 1K fake softwares, 50K fake comments, and 50K camouflage comments, is injected to the dataset. Specifically, we randomly generate the fake comments (resp. camouflage comments) between pairs of fake users and fake products (resp. real products). We note that this attack can be easily conducted in reality to help fake users evade the detection, e.g., fake users are coordinated to deliberately post comments on some real products~\cite{DBLP:conf/kdd/HooiSBSSF16}.  We then find MaxBPs and MBPs from the bipartite graph, and classify all users and products involved in the found subgraphs as fake items and others as real ones. 

We measure the running time and F1 score, and show the results in Figure~\ref{fig:case_study} where $\theta_L$ is fixed at 4. 
\begin{itemize}[leftmargin=*]
    \item 
    
    \textbf{Varying $\theta_R$}. 
    We find the 2000 MaxBPs (denoted by ``MaxBP (2000)''), first-2000 MBPs (denoted by ``MBP (2000)''), i.e., the first 2000 MBPs yielded by the algorithm, and all MBPs (denoted by ``MBP (All)'') with $\theta_R$ varying from 3 to 6, and show the results in Figure~\ref{fig:case_study}(a).
    We have the following observations. (1) MaxBP (2000) achieves the best F1 score (0.99) when $\theta_R = 5$ among all methods. (2) Both MBP (2000) and MBP (All) achieve their best F1 scores, 0.80 and 0.87, respectively, when $\theta_R = 5$. (3) Under the setting with the best F1 scores, i.e., $\theta_R = 5$, MaxBP (2000) and MBP (2000) run comparably fast and both of them run faster than MBP (All).
    For (1) and (2), the reason could be that the fraud blocks tend to reside in large MBPs with more edges, and for (3), the reason could be that MaxBP (2000) and MBP (2000) need to explore a set of \emph{some} but not all MBPs.
    
    \item \textbf{Varying $K$}. 
    We find the $K$ MaxBPs (denoted by ``MaxBP (K)'') and the first-$K$ MBPs (denoted by ``MBP (K)'') with $\theta_R=5$ since it give the best F1 scores, and show the results in Figure~\ref{fig:case_study}(b). We have the following observations.
    (1) MaxBP (K) has the F1 score higher than MBP (K) on all settings, and achieves the best F1 score at $K=2000$. 
    (2) MaxBP (K) provides a better trade-off between F1 score and running time than MBP (K), e.g., MaxBP (K) provides a F1 score 0.94 with the running time 512 seconds at $K = 1000$ while MBP (K) provides a similar F1 score 0.91 with the running time 1512 seconds at $K = 16000$.
\end{itemize}
In summary, MaxBP outperforms MBP for fraud detection in terms of F1 score and the running time, as shown in Figure~\ref{fig:case_study}.

{\roundB We have also explored methods of enumerating maximal bicliques and maximum bicliques for this case study. We 
found that the best F1-score achieved by enumerating maximal bicliques (with $\theta_R=4$) and that of finding the maximum bicliques (with $\theta_R=4$) are 0.48 and 0.54, respectively. In particular, when $\theta_R\geq 5$, there are few large bicliques, and thus both methods have the recall close to 0. This is mainly because biclique requires full connections among vertices, which is usually too strict.}

\begin{figure}[]
	\centering
	\vspace{-0.15in}
	\begin{tabular}{c c}
		
		\begin{minipage}{3.80cm}
			\includegraphics[width=4.2cm]{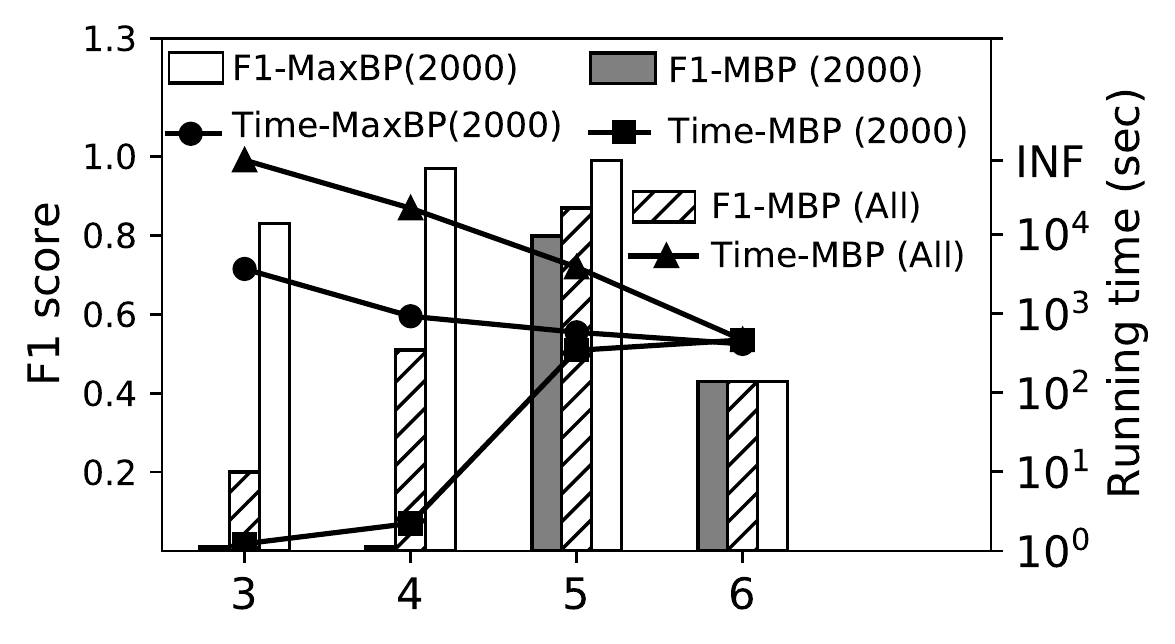}
		\end{minipage}
		&
		\begin{minipage}{3.80cm}
			\includegraphics[width=4.2cm]{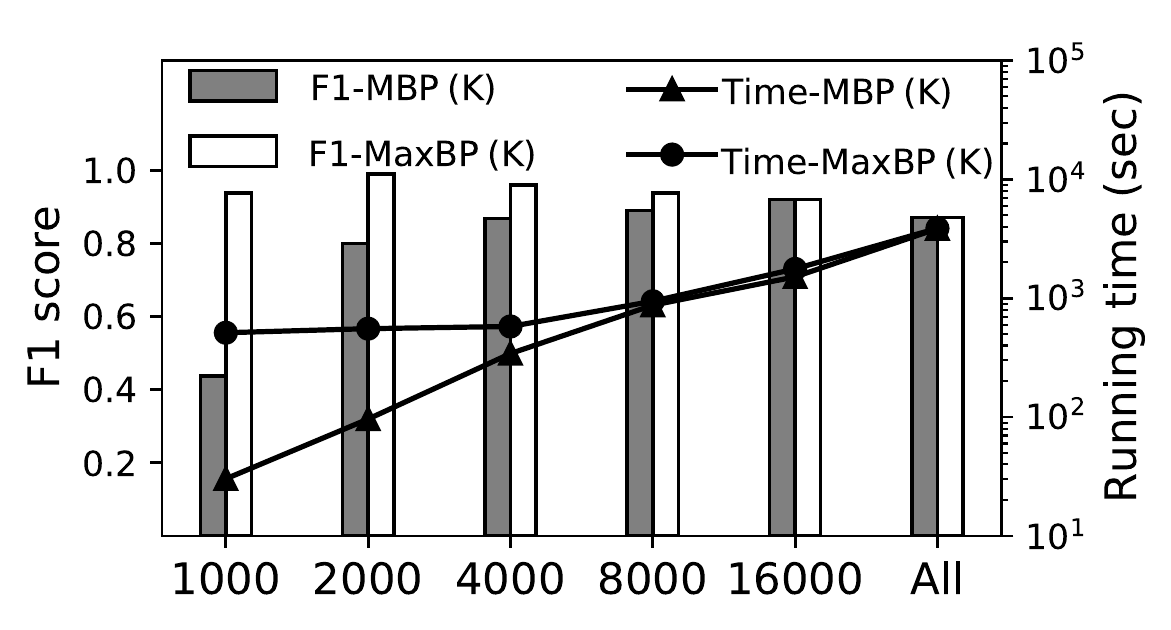}
		\end{minipage}		
		\\
		(a) Varying $\theta_R$
		&
		(b) Varying $K$ ($\theta_R=5$)
	\end{tabular}
 	\vspace{-0.15in}
	\caption{\roundA Case Study: Fraud Detection}
	\label{fig:case_study}
\end{figure}

\section{Related Work}
\label{sec:related}


{\LC We find that the maximum $k$-biplex search problem was recently studied in~\cite{luo2022maximum}. 
Some summary of the similarities and differences between \cite{luo2022maximum} and our paper include: (1) \cite{luo2022maximum} and our paper propose the maximum $k$-biplex search problem concurrently and independently; (2) both \cite{luo2022maximum} and our paper prove the NP-hardness of the problem, yet they use different NP complete problems; (3) the solution in \cite{luo2022maximum} has the time complexity of $O^*(2^{|V|})$ while 
our \texttt{FastBB} algorithm has the time complexity of $O^*(\gamma_k^{|V|})$ with $\gamma_k < 2$; 
and (4) the solutions in this paper run significantly faster than that in \cite{luo2022maximum}.
Some more details of the comparison can be found 
\ifx \CR\undefined
{\LC in the appendix.}
\else
{\LC in the technical report~\cite{TR}.}
\fi
} 
We review other related work as follows.

\smallskip
\noindent\textbf{Cohesive bipartite structures.}
Recently, many studies have been conducted on finding cohesive subgraphs of bipartite graphs, including bicliques~\cite{zhang2014finding,DBLP:conf/ijcai/AbidiZCL20,mukherjee2016enumerating,lyu2020maximum,chen2021efficient,kloster2019mining, Yangpqcounting}, $(\alpha,\beta)$-cores~\cite{liu2020efficient}, quasi-bicliques~\cite{DBLP:conf/cocoon/LiuLW08,wang2013near,ignatov2018mixed,DBLP:conf/icfca/Ignatov19}, $k$-biplexes~\cite{DBLP:journals/sadm/SimLGL09,yu2021efficient,yu2022kbiplex}, $k$-bitrusses~\cite{DBLP:conf/dasfaa/Zou16,wang2020efficient}, etc.
Biclique is a bipartite graph where any vertex at one side connects all vertices at the other side. Recent works on biclique focus on enumerating maximal bicliques~\cite{zhang2014finding,DBLP:conf/ijcai/AbidiZCL20,kloster2019mining}.
%
Given a bipartite graph, $(\alpha,\beta)$-core is the maximal subgraph where any vertex at one side connects a certain number of vertices (i.e., $\alpha$ or $\beta$) at the other side. It has many applications, including recommendation systems~\cite{ding2017efficient} and community search~\cite{DBLP:conf/ssdbm/HaoZW020,DBLP:journals/corr/abs-2011-08399}.
Existing works of $k$-biplexes focus on enumerating large maximal $k$-biplex~\cite{DBLP:journals/sadm/SimLGL09,yu2021efficient}.
A $k$-bitruss~\cite{DBLP:conf/dasfaa/Zou16,wang2020efficient} is a bipartite graph where each edge is contained in at least $k$ butterflies,
where a butterfly corresponds to a complete 2$\times$2 biclique~\cite{DBLP:journals/pvldb/WangLQZZ19}.
In the literature, there are two types of quasi-biclique, i.e., (1) $\delta$-quasi-biclique~\cite{DBLP:conf/cocoon/LiuLW08} is a bipartite graph $G(L,R,E)$ where each vertex in $L$ (resp. $R$) misses at most $\delta \cdot |R|$ (resp. $\delta \cdot |L|$) edges with $\delta\in [0,1)$) and {\roundB (2) $\gamma$-quasi-biclique~\cite{ignatov2018mixed,yan2005identifying}} is a bipartite graph $G(L,R,E)$ that can miss at most $\gamma \cdot |L|\cdot|R|$ edges with $\gamma\!\in\! [0,1)$.
Existing works of quasi-bicliques focus on finding subgraphs with a certain density and degree~\cite{DBLP:conf/kdd/MitzenmacherPPT15,DBLP:conf/bigdata/LiuSC13}.
%
In this paper, we focus on $k$-biplex since it (1) imposes strict enough requirements on connections within a subgraph and tolerates some disconnections and (2) satisfies the hereditary property, which facilitates efficient solutions. 
In~\cite{yu2022kbiplex}, a case study of fraud detection on e-commerce platforms is conducted, which shows that $k$-biplex works better than some other cohesive subgraph structures including biclique, $(\alpha,\beta)$-core, and $\delta$-quasi-biclique for the application.

\smallskip
\noindent\textbf{Maximum biclique search.} 
The maximum biclique search problem has attracted much attention in recent years~\cite{dawande2001bipartite,sozdinler2018finding,lyu2020maximum,shaham2016finding,garey1979computers,DBLP:conf/cpaior/McCreeshP14,zhou2018towards,chen2021efficient,wang2018new,DBLP:journals/tcyb/YuanLCY15,DBLP:journals/corr/ZhouH17a}. In general, there are three lines of works, namely maximum edge biclique search (MEBS)~\cite{dawande2001bipartite,sozdinler2018finding,lyu2020maximum,shaham2016finding} which finds a biclique $H^*$ such that $E(H^*)$ is maximized, maximum vertex biclique search (MVBS)~\cite{garey1979computers} which finds a biclique $H^*$ such that $V(H^*)$ is maximized and maximum balanced biclique search (MBBS)~\cite{DBLP:conf/cpaior/McCreeshP14,zhou2018towards,chen2021efficient,wang2018new,DBLP:journals/tcyb/YuanLCY15,DBLP:journals/corr/ZhouH17a} which finds a biclique $H^*$ such that $V(H^*)$ is maximized and $L(H^*)=R(H^*)$. 
First, the MEBS problem is NP-hard, for which many techniques have been proposed. Authors in \cite{dawande2001bipartite,sozdinler2018finding} adopt the integer linear programming techniques to find a MEB, which is not scalable for large bipartite graphs. A recent study~\cite{lyu2020maximum} proposes a progressive bounding framework to deal with large bipartite graphs.  Besides, a Monte Carlo algorithm is proposed in \cite{shaham2016finding}, which finds a MEB with a fixed probability. 
Second, the MVBS problem can be solved in polynomial time by finding a maximum matching~\cite{garey1979computers}.
Third, the MBBS problem is NP-hard, for which both exact methods and approximate methods have been developed. To be specific, exact methods proposed in \cite{DBLP:conf/cpaior/McCreeshP14,zhou2018towards,chen2021efficient} {\roundA are branch-and-bound algorithms which use the widely-used Bron-Kerbosch (BK) branching~\cite{bron1973algorithm}.}
Besides, approximate methods include \cite{wang2018new,DBLP:journals/tcyb/YuanLCY15} which introduce a local search framework to find an approximate MBB, \cite{DBLP:journals/tcas/Al-YamaniRP07,DBLP:journals/jetc/Tahoori06} which {\LC convert MBBS into a maximum independent set problem and adopt approximate algorithms for the independent set problem}, and \cite{DBLP:journals/corr/ZhouH17a} which proposes a heuristic algorithm with tabu search and graph reduction. 
{\roundA In summary, there are two types of solutions, namely exact methods and approximate methods. For exact algorithms, most of them follow the branch-and-bound framework and use the BK branching strategy. However, based on our experimental results, BK branching strategy performs worse than our proposed Sym-BK branching strategy for the problem of finding MaxBP.
For approximate methods, they cannot be adapted to find the MaxBP exactly.}

{\roundA
\smallskip
\noindent\textbf{Maximum quasi-biclique search.} The maximum quasi-biclique search problem aims to find a $\gamma$-quasi-biclique or $\delta$-quasi-biclique $H^*$ {\LC such that} $V(H^*)$ is maximized, which is a NP-hard problem~\cite{DBLP:conf/cocoon/LiuLW08,DBLP:conf/icfca/Ignatov19}. 
\cite{DBLP:conf/icfca/Ignatov19,ignatov2018mixed} use mixed integer programming to find a maximum $\gamma$-quasi-biclique exactly, which cannot handle large datasets. 
\cite{wang2013near,DBLP:conf/cocoon/LiuLW08} propose greedy algorithms to find an approximate maximum $\delta$-quasi-biclique, which cannot be adapted to find the MaxBP exactly. 
}
\balance
\section{Conclusion}
\label{sec:conclusion}
In this paper, we study the \emph{maximum $k$-biplex search} problem, which is to find $K$ maximal $k$-biplexes with the most edges. 
We propose two branch-and-bound algorithms, among which the better one \texttt{FastBB} is based on a novel Sym-BK branching strategy and achieves better worst-case time complexity than adaptions of existing algorithms. We further develop frameworks to boost the efficiency and scalability of the branch-and-bound algorithms including \texttt{FastBB}.
Extensive experiments are conducted on real and synthetic datasets to demonstrate the efficiency of our algorithms and the effectiveness of proposed techniques. 
In the future, we plan to develop efficient parallel algorithms for the maximum $k$-biplex search problem and explore the possibility of adapting our algorithms to find other types of maximum cohesive subgraphs in bipartite graphs.

\section{Proof of Lemma~\ref{lemma:np}}
\label{sec:proof}
We prove by showing a polynomial reduction from a well-known NP-complete problem, namely \emph{maximum clique search}, to the maximum $k$-biplex search problem {\LC with $\theta_L=\theta_R=0$, $k = 1$, and $K = 1$}. 
We define their decision problems as follows. 
\begin{itemize}[leftmargin=*]
    \item CLIQUE: given a general graph $G=(V,E)$ and a positive integer $\alpha$, does $G$ contain a clique with at least $\alpha$ vertices?
    \item BIPLEX: given a bipartite graph $\mathcal{G}=(\mathcal{L}\cup \mathcal{R},\mathcal{E})$ and two positive integers $k$ and $\alpha'$, does $\mathcal{G}$ contain a $k$-biplex with at least $\alpha'$ edges?
\end{itemize}

Let $G=(V,E)$ and $\alpha$ be the inputs of an instance of {CLIQUE}. W.l.o.g., we assume that $\alpha=\frac{1}{2} |V|$ is a positive integer. This can be achieved with some inflation tricks. Specifically, (1) if the original input $\alpha$ is smaller than $\frac{1}{2}|V|$ (which can be a fractional number), we add to $G$ a new vertex $v$ and $|V|$ edges between $v$ and other vertices. Obviously, every clique would include $v$, increasing $\alpha$ by 1 but $\frac{1}{2}|V|$ by 0.5. (2) If $\alpha>\frac{1}{2}|V|$, we add to $G$ a new vertex $v$ with no edges. Clearly, every clique would not include $v$, increasing $\frac{1}{2}|V|$ by 0.5 only. By repeating above two steps, we can make $\alpha=\frac{1}{2} |V|$ finally.

Now, we construct an instance of BIPLEX (with $k=1$) with $\mathcal{G}=(\mathcal{L}\cup\mathcal{R},\mathcal{E})$ and $\alpha'$. To be specific, 
\begin{equation}
    \mathcal{L}=V \text{ and }  \mathcal{R}=E\cup W\cup U,\  |W|= \frac{1}{2}\alpha(\alpha-5),\  |U|=2\alpha \nonumber
\end{equation}
where $W\cup U$ is a set of new elements and $|V|=2\alpha$. Assume $V=\{v_1,...,v_{2\alpha}\}$ and $U=\{u_1,...,u_{2\alpha}\}$. We have
\begin{gather}
    \mathcal{E}=\mathcal{E}(\mathcal{G}[V\cup E])\cup \mathcal{E}(\mathcal{G}[V\cup W])\cup \mathcal{E}(\mathcal{G}[V\cup U])\nonumber\\
    =\{(v,e)\mid v\in V, e\in E, v\notin e\}\cup \{(v,w)\mid v\in V, w\in W\} \nonumber\\
    \cup \{(v_i,u_j)\mid v_i\in V, u_j\in U, i\neq j\},\nonumber\\
    \alpha'=\frac{1}{2}\alpha^3+\frac{3}{2}\alpha^2-\alpha\nonumber.
\end{gather}
To guarantee $|W|\geq 0$, we assume $\alpha\geq 5$ {\LC which can be achieved} with the inflation tricks. The above construction can be finished in polynomial time.


We then show that \emph{$G$ has a clique with at least $\alpha$ vertices iff $\mathcal{G}$ has a 1-biplex with at least $\alpha'$ edges}. 
We consider the following two cases.

\noindent\textbf{Case 1:} $G$ has a clique $G[C]$ with $\alpha$ vertices, i.e., $C\subseteq V$ and $|C|=\alpha$. Consider $X=V\backslash C$ and $Y=E(G[C])\cup W\cup U$. We prove this case by showing that $\mathcal{G}[X\cup Y]$ is a 1-biplex with $\alpha'$ edges. To be specific, for each vertex $v \in X$, we know: (1) vertex $v$ connects all vertices from $E(G[C])$ since $v$ is not an endpoint of any edge in $E(G[C])$, thereby yielding $|X|\times |E(G[C])|$ edges in total; (2) vertex $v$ connects all vertices from $W$ based on the construction, yielding $|X|\times |W|$ edges in total; and (3) vertex $v$ disconnects exactly one vertex in $U$ based on the construction, yielding $|X|\times (|U|-1)$ edges in total. For each vertex $u\in Y$, we can similarly verify that vertex $u$ disconnects no more than one vertex from $X$. In summary, we have a 1-biplex $\mathcal{G}[X\cup Y]$ with
\begin{equation}
    |\mathcal{E}(\mathcal{G}[X\cup Y])|=|X|\times |E(G[C])|+|X|\times |W|+|X|\times (|U|-1),\nonumber
\end{equation}
where $|X|=\alpha$, $|E(G[C])|=\frac{1}{2}\alpha(\alpha-1)$, $|W|=\frac{1}{2}\alpha(\alpha-5)$ and $|U|=2\alpha$. Therefore, $|\mathcal{E}(\mathcal{G}[X\cup Y])|
$ is exactly $\alpha'$.

\smallskip
\noindent\textbf{Case 2:} $G$ does not have a clique with at least $\alpha$ vertices. If no 1-biplex is found in $\mathcal{G}$, the proof is finished. Otherwise, let $\mathcal{G}[X\cup Y]$ be an arbitrary 1-biplex in $\mathcal{G}$ such that $X\subseteq \mathcal{L}$ and $Y\subseteq \mathcal{R}$. We finish the proof by showing that $|\mathcal{E}(\mathcal{G}[X\cup Y])|< \alpha'$. To estimate $|\mathcal{E}(\mathcal{G}[X\cup Y])|$, we divide $Y$ into two disjoint parts $Y_0\cup Y_1$, i.e., vertices in $Y_0$ connect all vertices from $X$ and vertices in $Y_1$ disconnect exactly one vertex from $X$. For $Y_0$, we have: (1) it includes all vertices from $E(G[V\backslash X])\subseteq E$ since every edge in $E(G[V\backslash X])$ has no endpoint in $X$; (2) it includes all vertices from $W$ based on our construction; and (3) it includes $|V\backslash X|$ vertices from $U$ (with $X=\{v_1,...,v_{|X|}\}$, vertices in $\{u_{|X|+1},...,u_{2\alpha}\}$ would connect all vertices from $X$ based on our construction). For $Y_1$, it includes at most $|X|$ vertices since otherwise there exists at least a vertex in $X$ that disconnects at least 2 vertices based on the pigeonhole principle. This contradicts to the definition of 1-biplex $\mathcal{G}[X\cup Y]$. In summary, we have
\begin{eqnarray}
\label{eq:edge_estimator}
    |\mathcal{E}(\mathcal{G}[X\cup Y])|= |X|\times |Y_0|+(|X|-1)\times |Y_1|\ \ \ \ \ \ \nonumber\\ 
    \leq|X|\!\times\! (|E(G[V\backslash X])|\!+\!|W|\!+\!|V\backslash X|)+(|X|\!-\!1)\!\times\! |X|
\end{eqnarray}
Let $x=|X|$. We have $0\leq x\leq 2\alpha$ and consider two cases.

\noindent \emph{\textbf{Case 2.1:}} $\alpha<x\leq 2\alpha$. Let $y= x-\alpha$ ($0<y\leq \alpha$). We have $|X|=\alpha+y$ and $|V\backslash X|=\alpha-y$. Moreover, we can obtain $|E(G[V\backslash X])|\leq \frac{1}{2} |V\backslash X|(|V\backslash X|-1)=\frac{1}{2}(\alpha-y)(\alpha-y-1)$ where the equality holds iff $G[V\backslash X]$ is a clique. According to equation (\ref{eq:edge_estimator}), we have $|\mathcal{E}(\mathcal{G}[X\cup Y])|\leq (\alpha+y)[\frac{1}{2}(\alpha-y)(\alpha-y-1)+\frac{1}{2}\alpha^2-\frac{5}{2}\alpha+2\alpha-(\alpha+y)]+(\alpha+y-1)(\alpha+y)$ which reduces to
\begin{equation}
    |\mathcal{E}(\mathcal{G}[X\cup Y])|-\alpha'\leq \frac{1}{2} y[y^2-(\alpha-1)y-\alpha-2].\nonumber
\end{equation}
It is easy to verify that $y^2-(\alpha-1)y-\alpha-2$ is negative for $0\leq y\leq \alpha$. Therefore, we have $|\mathcal{E}(\mathcal{G}[X\cup Y])|<\alpha'$.

\noindent\emph{\textbf{Case 2.2:}} $0\leq x\leq \alpha$. Let $y=\alpha-x$ ($0\leq y\leq \alpha$). We have $|X|=\alpha-y$ and $|V\backslash X|=\alpha+y$. Since $|V\backslash X|=2\alpha-x \geq \alpha$ and $G$ does not have a clique with at least $\alpha$ vertices, it is easy to verify that $|E(G[V\backslash X])|< \frac{1}{2}|V\backslash X|(|V\backslash X|-1)-y$ since otherwise there exists a clique with at least $\alpha$ vertices in $G$ (note that the right term can be regarded as a process of iteratively removing $y$ edges from a clique with $|V\backslash X|$ vertices and after removing an edge, the maximum clique in the remaining graph has its size decrease by at most 1, which yields a clique with $|V\backslash X|-y=\alpha+y-y=\alpha$ vertices). 
According to equation (\ref{eq:edge_estimator}), we have $|\mathcal{E}(\mathcal{G}[X\cup Y])|< (\alpha-y)[\frac{1}{2}(\alpha+y)(\alpha+y-1)-y+\frac{1}{2}\alpha^2-\frac{5}{2}\alpha+2\alpha-(\alpha-y)]+(\alpha-y-1)(\alpha-y)$ which reduces to
\begin{equation}
    |\mathcal{E}(\mathcal{G}[X\cup Y])|-\alpha'< \frac{1}{2} y[-y^2-(\alpha-3)y-\alpha+2].\nonumber
\end{equation}
It is easy to verify that $y[-y^2-(\alpha-3)y-\alpha+2]\leq 0$ for $0\leq y\leq \alpha$ and $\alpha\geq 5$. We thus have $|\mathcal{E}(\mathcal{G}[X\cup Y])|<\alpha'$.
\section{Acknowledgement}
\label{sec:ack}
The research is supported by the Ministry of Education, Singapore, under its Academic Research Fund (Tier 2 Award MOE-T2EP20221-0013 and Tier 1 Award (RG77/21)). Any opinions, findings and conclusions or recommendations expressed in this material are those of the author(s) and do not reflect the views of the Ministry of Education, Singapore.
\balance
\clearpage
\bibliographystyle{ACM-Reference-Format}
\bibliography{SIGMOD_MaxBP}

\ifx \CR\undefined
 \clearpage
\input{appendix}
\balance
\fi

\end{document}